\documentclass[12pt]{article} 
\usepackage[ruled]{algorithm2e}

\usepackage{amssymb, amsthm, amsmath, color}
\usepackage[margin=1.0in]{geometry}

\newtheorem{assumption}{Assumption}[section]
\newtheorem{definition}{Definition}[section]
\newtheorem{theorem}{Theorem}[section]
\newtheorem{proposition}{Proposition}[section]

\usepackage{tkz-graph}

\newcount\Comments  
\newcommand{\kibitz}[2]{\ifnum\Comments=1{\color{#1}{#2}}\fi}

\usepackage{authblk}
\usepackage{yfonts}

\newcommand{\m}[1]{\mathcal{#1}}
\renewcommand{\b}[1]{\boldsymbol{#1}}
\newcommand{\doublespace}{\hspace{2px}}
\newcommand{\newexample}[2]{
\vspace{5px} \noindent \textbf{#1.}
#2
\newline \noindent}

\newcommand{\Reals}{\mathbb{R}}
\newcommand{\units}{\mathcal{U}}
\newcommand{\maxUnit}{m}
\newcommand{\agents}{\mathcal{I}}
\newcommand{\outcomes}{\Reals}

\newcommand{\UperA}{k}
\newcommand{\agent}{i}
\newcommand{\maxAgent}{n}

\newcommand{\unit}{u}

\newcommand{\Yobs}[1]{Y_{#1}^{\mathrm{obs}}}
\newcommand{\Yobsbar}[1]{\overline{\Yobs{#1}}}

\newcommand{\block}{b}
\newcommand{\maxBlock}{B}
\newcommand{\actions}{\mathcal{A}}

\newcommand{\design}{\m{D}}

\newcommand{\winner}{\hat{\tau}(\Yobs{..})}
\newcommand{\Ex}[1]{\mathbb{E}(#1)}
\newcommand{\Var}[1]{\mathbb{V}\mathrm{ar}(#1)}

\newcommand{\Sa}{\Sigma(\A)}
\newcommand{\N}[1]{\mathcal{N}(0, #1)}

\newcommand{\chiA}{\b{\chi}(\A)}
\newcommand{\phiB}{\b{\phi}(\Yobs{..})}
\newcommand{\designTuple}{(\psi, \phi)}

\newcommand{\tod}{\xrightarrow{D}}

\newcommand{\Z}{\mathbf{Z}}
\newcommand{\A}{\mathbf{A}}

\usepackage{mathtools}
\newcommand\defeq{\stackrel{\mathclap{\normalfont\mbox{\scriptsize def}}}{=}}

\newcommand{\ExCond}[2]{\mathbb{E}\left(#1|#2\right)}

\newcommand{\lam}[1]{\lambda_{#1}}
\newcommand{\lamc}[1]{\lambda^{'}_{#1}}
\newtheorem{innercustomthm}{Theorem}
\newenvironment{customthm}[1]
  {\renewcommand\theinnercustomthm{#1}\innercustomthm}
  {\endinnercustomthm}

\newcommand{\drawMone}[2]{
\begin{figure}[h]
\centering
\begin{tikzpicture}[scale=1.0]
\draw (#1, 0)  node[draw,shape=circle, color=gray, text=black, fill=white, inner sep=0.7mm] (s1) {\footnotesize seed 1};
\draw (0, #2)  node[draw,shape=circle, color=gray, text=black, fill=white, inner sep=3.5mm] (b11) {\footnotesize  Test set 1};
\draw (5 * #1, 0)  node[draw,shape=circle, color=gray, text=black, fill=white, inner sep=0.7mm] (s2) {\footnotesize seed 2};
\draw  (6 * #1, #2)  node[draw,shape=circle, color=gray, text=black, fill=white, inner sep=3.5mm] (b22) {\footnotesize  Test set 2};
\draw [dashed] (3 * #1, #2+0.8) -- (3 * #1, -.5);
 \draw [->] (s1) -- (b11);
 \draw [->, dashed] (s1) -- (b22);
 \draw [->] (s2) -- (b22);
 \draw [->, dashed] (s2) -- (b11);
\node[text width=4cm] at (1.5 * #1, 0.4 * #2) {$\lam{1}$};
\node[text width=4cm] at (5.5 * #1, 0.4 * #2) {\textcolor{white}{............}$\lam{2}$};
\node[text width=4cm] at (3.5 * #1, 0.8 * #2) {$\gamma \lamc{2}$};
\node[text width=4cm] at (5.2 * #1,  0.8 * #2) {$\gamma \lamc{1}$};
\end{tikzpicture}
\caption{Test set $i$ has units assigned to agent $i$, i.e., 
$\{u \in \units : Z_u=i\}$.
Seed set $i$ corresponds to the treatment version $A_i$.
The seed sets influence the purchase rate of units in the test sets,
for example, through word-of-mouth effects between units. In particular, $A_i = (\lam{i},\lamc{i})$, where $\lam{i}$ 
is the induced rate from seed set $i$ to test set $i$, 
and $\gamma \lamc{i}$ is the induced rate from 
seed set $i$ to the other
test set, where $0\le\gamma\le1$ is a parameter 
that models interference.
Outcomes, i.e., product purchases, are measured on units in the test sets; the score of agent $i$ will be calculated based on 
observed purchases in test set $i$. 
Arrows indicate induced purchase rates from the seed sets; dashed arrows indicate that the rate is discounted by $\gamma$. 
The presence of interference, where an agent can affect 
the purchase rate on a test set of another agent, changes how 
agent select their seed sets, i.e., their treatment versions.
}
\label{fig:m1}
\end{figure}}

\newcommand{\drawMtwo}[2]{
\begin{figure}[h]
\centering
\begin{tikzpicture}[scale=1.0]
\draw (-1, 0)  node[draw,shape=circle, color=white, text=black, fill=white, inner sep=0.7mm] (s122) {\textcolor{white}{aasasasas}};

\draw (#1, 0)  node[draw,shape=circle, color=gray, text=black, fill=white, inner sep=0.7mm] (s1) {\footnotesize seed 1};
\draw (0, #2)  node[draw,shape=circle, color=gray, text=black, fill=white, inner sep=1mm] (b11) {\footnotesize  Test set $G_{11}$ };
\draw (2 * #1, #2)  node[draw,shape=circle, color=gray, text=black, fill=white, inner sep=1mm] (b12) 
{\footnotesize Test set $G_{12}$};
\filldraw[fill opacity=0] 
(-0.8* #1, #2-1.2) rectangle (2.8 * #1, 1.5 * #2);
\node[text width=4cm] at (#1, 1.6 * #2) {
Group $G_1$};
\draw (5 * #1, 0)  node[draw,shape=circle, color=gray, text=black, fill=white, inner sep=0.7mm] (s2) {\footnotesize seed 2};
\draw (4 * #1, #2)  node[draw,shape=circle, color=gray, text=black, fill=white, inner sep=1mm] (b21) 
{\footnotesize Test set $G_{21}$};
\draw (6 * #1, #2)  node[draw,shape=circle, color=gray, text=black, fill=white, inner sep=1mm] (b22) {\footnotesize  Test set $G_{22}$};
\filldraw[fill opacity=0] 
(3.2 * #1, #2-1.2) rectangle (6.8 *#1, 1.5 * #2);
\node[text width=4cm] at (6.5 * #1, 1.6 * #2) {
Group $G_2$};
\draw [dashed] (3 * #1, 1.8 * #2) -- (3 * #1, -.5);
 \draw [->] (s1) -- (b11);
 \draw [->, dashed] (s1) -- (b12);
 \draw [->] (s1) -- (b21);
 \draw [->, dashed] (s1) -- (b22);
 \draw [->] (s2) -- (b22);
 \draw [->, dashed] (s2) -- (b21);
 \draw [->] (s2) -- (b12);
 \draw [->, dashed] (s2) -- (b11);
\end{tikzpicture}
\caption{
Test sets $G_{1j}$ and $G_{2j}$
 have the units assigned to agent $j$, i.e., 
$\{u \in \units : Z_u=i\}$; there are two test sets per agent.
Agent $i$ selects an influential seed set $i$, that corresponds to the treatment version $A_i$. The seed sets influence the purchase rate of units in the test sets. In particular, $A_i = (\lam{i},\lamc{i})$, where $\lam{i}$ 
is the induced rate from seed set $i$ to a test set 
with units assigned to $i$, and $\gamma \lamc{i}$ is the induced rate from seed set $i$ to a test set with units 
assigned to the other agent. Outcomes are measured on units in the test sets;
the score of agent $i$ will be calculated based on 
observed purchases of units assigned to agent $i$; 
for example, agent 1 will be scored based on outcomes 
of units in $G_{11}$ and $G_{21}$. Arrows indicate induced purchase rates from the seed sets; dashed arrows indicate that the rate is discounted by $\gamma$. 
Agent scores are calculated based on outcomes in their respective test sets. The presence of interference, where an agent can affect 
the purchase rate on a test set of another agent, changes how 
agent select their seed sets, i.e., their treatment versions.
}
\label{fig:m2}
\end{figure}}

\newcommand{\ones}{\mathbf{1}}

\begin{document}

\markboth{P. Toulis et al.}{Incentive-Compatible Experimental Design}

\title{Incentive-Compatible Experimental Design}

\author[$\alpha$]{Panos Toulis}
\author[$\beta$]{David C. Parkes}
\affil[$\alpha$]{\small Department of Statistics, Harvard University}
\affil[$\beta$]{\small School of Engineering and Applied Science, Harvard University}
\author[$\beta$]{Elery Pfeffer}
\author[$\gamma$]{James Zou}
\affil[$\gamma$]{Microsoft Research}

\maketitle

\begin{abstract}
  We consider the design of experiments to evaluate treatments that
  are administered by self-interested agents, each seeking to achieve
  the highest evaluation and win the experiment. For example, in an
  advertising experiment, a company wishes to evaluate two 
marketing agents in
  terms of their efficacy in viral marketing, and assign a contract to
  the winner agent.  Contrary to traditional experimental design, this
  problem has two new implications.  First, the experiment induces a
  game among agents, where each agent can select from multiple
  versions of the treatment it administers. Second, the action of one
  agent -- selection of {\em treatment version} -- may affect the
  actions of another agent, with the resulting
  \emph{strategic interference} complicating the evaluation of agents.
  An {\em incentive-compatible experiment design} is one with an
  equilibrium where each agent selects its {\em natural action}, which
  is the action that would maximize the performance of the agent if there was no competition (e.g., expected number of conversions if agent was assigned the contract).

Under a general formulation of experimental design,
we identify sufficient conditions that guarantee incentive-compatible experiments.
These conditions rely on the existence of statistics that 
can estimate how agents would perform without competition,
and their use in constructing score functions to evaluate the agents. In the setting with no strategic interference, 
we also study the \emph{power} of the design, i.e., the probability that the best agent wins, and show how to improve 
the power of incentive-compatible designs.
From the technical side, our theory uses  a range of statistical methods 
such as hypothesis testing, variance-stabilizing transformations and the Delta method, all of which rely on asymptotics.
\end{abstract}


\section{Introduction}
\label{section:introduction}

Experiments are the gold-standard for evaluating the effects of
different treatments. The design of experiments is crucial in order to avoid
systematic biases and to minimize random errors in the statistical
evaluation of treatment effects~\cite{cox2000theory}.  There 
are three fundamental concepts in any experiment design. 
The \emph{treatment} is a well-defined prescription or set of rules, e.g., 
a pharmaceutical drug, a marketing campaign, or a new material.
The goal of the experiment is to evaluate the effects of different 
 treatments. 
The \emph{experimental unit} is the indivisible entity that will
receive a treatment within the experiment, e.g., a patient, a
potential customer, or a factory process. Typically, every unit
receives only one treatment, but there are important exceptions as
well.  The treatment is assigned according to a \emph{treatment
  assignment rule} specified by the design and necessarily involves
randomization in order to avoid systematic biases.  When a unit
receives the treatment it exhibits a measurable \emph{outcome}, e.g., a health assessment, a product purchase or not, or a material failure rate.

Statistical analysis of unit outcomes is necessary for the evaluation
of treatments because it accounts for the errors that are inherent to
randomization of treatment and the measurement process.
A key idea in experimental design is \emph{blocking}. Background information 
on units is almost always available, e.g., age, gender, socioeconomic status, 
health status, and so on. If an experimenter believes that 
units' outcomes vary systematically with respect to such \emph{covariate} information, then it is necessary to block units with respect to the 
available covariates.
Blocking helps to avoid systematic bias and variability that is not of
scientific interest. The unofficial mantra in experimental design is
``block what you can and randomize what you cannot"
Box et. al. \cite{box1978statistics}).

To illustrate, consider the example of a new flu shot.  A
pharmaceutical company, the \emph{experimenter}, wants to compare
between the new flu shot and a baseline that is currently in the
market. The treatments are the two flu shots.  The experimenter
has a set of volunteer patients who form the set of experimental
units. When a unit receives a treatment the outcome is whether the
unit got flu or not for the three months following the treatment.  As a
treatment assignment rule, the experimenter could simply give the new
flu shot to half of the patients at random, and give the baseline to
the other half.  However, the outcomes could be confounded with
factors such as age (older people are more vulnerable to flu),
geography (urban areas are more crowded and possibly more contagious),
occupation, and so on.  In a blocking design, the experimenter could
block the population based on age and occupation, and perform the
randomization within blocks.

There are two crucial assumptions in experimental design
and the related topic of {\em causal inference}, collectively known as
the {\em stable unit treatment value assumption}
(SUTVA)~\cite{rubin1980comment}.  First, there are \emph{no hidden
  versions} of a treatment. In the previous example, this means that
there are no strong or weak versions of the new flu shot.  Otherwise,
the outcomes would be confounded with the hidden version of the
treatment.  This is an important problem, especially in social
science studies.  For example, in an educational study a new treatment
could be a new type of curriculum, however a possible hidden version
of the treatment is the delivery method by each teacher.  A second
crucial assumption is that of \emph{no interference} among
experimental units. Interference is present when the treatment
assignment on one unit affects the outcome of another unit. In the flu
shot example, a unit that is not vaccinated is still protected when
the friends of the unit are vaccinated. Neither of these assumptions hold in our
setting.

We introduce the idea of {\em incentive-compatible experimental
  design} in the context of viral marketing.\footnote{An early
  extended abstract of this paper was presented in the Conference on
  Digital Experimentation at MIT Toulis et. al.\cite{toulisCODE14}.} Imagine a company
that designs a test to determine which of two vendors has the best
algorithm for running an advertising campaign. The firm uses
randomization to prevent systematic bias, and defines a criterion for
success; e.g., the number of conversions over a two week period. The
winning vendor is promised a one-year contract with the firm running
the test. One challenge in this setting is that the vendors might
deviate from how they would normally run a campaign, trying to win the
test.  For example, a lower quality vendor may try to follow a more
aggressive strategy, hoping to get lucky.  This is a problem for the
firm designing the test, who wants to get an unbiased estimate of the
usual performance of the vendor.  Another challenge comes from
interference between the participants. In viral marketing, for
example, one vendor may try to free-ride on word-of-mouth effects that
come from another vendor.

\subsection{Results}

A first contribution of the present paper is to formalize this problem of incentive-compatible experimental design. 
The difference with traditional experimental design is that, 
in our framework, strategic agents administer the treatments 
to be evaluated, and each agent can select from multiple treatment
versions. In this way, the experiment induces a non-cooperative game. 
The action available to an agent in the resulting {\em treatment selection game} 
is the version of the treatment that 
the agent will administer to its assigned units.
The experimenter has a {\em performance metric} 
 to evaluate each treatment version. This is the quantity of interest
to the experimenter. 
Each agent has a {\em natural action}, which is the action that maximizes its performance,
and is assumed to be the way the agent would act if not competing in the game. The {\em quality} of an agent is the maximum value 
of the performance metric, achieved when the agent plays the natural action without competition from other agents.
The goal of the experimenter is to design an experiment to estimate 
the agent of highest quality.
 An {\em incentive-compatible experiment design} is one with 
an equilibrium
in which each agent's best response is to select the treatment version
corresponding to its natural action. We will focus on
dominant-strategy equilibrium in this paper.

We show that incentive-compatible designs are possible when an
\emph{identifying statistic} exists that can estimate the quality
difference between  agents (Theorem \ref{theorem:general}).
Critically, the variance of such a statistic has to be less sensitive to
agent actions than its expected value, otherwise an agent can take
advantage of the variance of the statistic.
Under a no interference assumption, a class of incentive-compatible designs 
can be constructed through a \emph{variance-stabilizing} transformation (Theorem \ref{theorem:no_interference}), which makes the 
variance of the identifying statistic insensitive to agent actions; 
a worse agent cannot hope to increase its chances 
by being more aggressive. 
This leads to results that may sound counter-intuitive.  For example,
in a viral marketing application where performance is the expected
number of conversions, and where higher expected conversions also
correspond to increasingly higher risks, it is not
incentive-compatible to select as the winner the agent with the
highest average performance; rather, it is incentive-compatible to
select as the winner the agent with the lowest reciprocal of average
performance (see Example~2(d)).

Identifying statistics and incentive-compatible designs are generally
harder to obtain under strategic interference. However, under specific
modeling assumptions about the interference, better designs can yield
more information about the agent performances, and thus produce
identifying statistics.  We illustrate this idea in a viral marketing
example, which we reuse throughout this paper.

\section{Preliminaries}

In this section we introduce notation for the operational and
statistical components of incentive-compatible experimental design.
The operational components include the {\em treatment assignment}, the
{\em treatment selection game} and the experiment {\em outcomes}.  The
statistical components include the {\em estimand}-- the quantity of
interest to the experimenter --and the {\em estimators}, i.e., the
data statistics used to estimate the estimand.

\subsection{Treatment assignment}
\label{section:notation:assignment}

Let $\units = \{1, 2, \ldots, \maxUnit \}$ denote the set of {\em
  experimental units}, indexed by $u$, and $\agents = \{1, 2, \ldots,
\maxAgent \}$ denote the set of {\em agents}, indexed by $i$.  Each
agent, for example, a marketing firm or a drug company, represents a
treatment to be evaluated.  An experimenter needs to design the
experiment that will evaluate the agents.  Relative to traditional
experimental design, the new aspect is that each agent is associated
with a set of {\em treatment versions} and each agent has a strategic
choice about which version to administer in the experiment. We make
this precise in Section~\ref{sec:game}.
 
For each unit $u \in \units$ there is covariate information that is
common knowledge to agents and the experimenter.  We assume the
experimenter uses covariates to split units into blocks,
such that units within one block are
similar in terms of covariates, e.g., similar age, gender, income,
etc. Without loss of generality, we will assume there is just a single
block. In Appendix \ref{section:multiple_blocks} of this paper, we discuss how the theory
can be extended to multiple blocks.

A {\em treatment assignment rule} $\psi$ assigns each unit to a single
agent. Let $\Z = (Z_u)$ denote the $\maxUnit \times 1$ assignment
vector, such that $Z_u=i$ indicates that unit $u$ is assigned to agent
$i$.  The assignment rule $\psi$ is a probability distribution over
all possible assignments $\Z$. Without loss of generality, we assume
that the number of units $\maxUnit$ is a multiple of the number of
agents $\maxAgent$. We will also assume complete randomization, such
that $Z_u=i$, for exactly $\UperA \defeq \maxUnit / \maxAgent$ units,
for each agent $i$.

\subsection{Treatment selection game}
\label{sec:game}

The set of actions $\actions_i\subseteq \actions$ denotes the feasible
action space for agent $i$, where $\actions$ is the set of all
possible actions.
Subsequent to treatment assignment, every agent $i$ 
simultaneously selects an action
$A_i \in \actions_i$, which corresponds to a version of the treatment
administered by agent $i$. The same version is applied
to all units assigned to agent $i$.\footnote{In Appendix 
\ref{section:multiple_blocks}, we introduce multiple blocks and allow an
  agent to pick a different action for each block. All units within a block receive the
  same treatment version, but versions might differ across blocks.}
Let $\A = (A_1, \ldots A_\maxAgent)$ denote the joint action profile, and
$\A_{-i} = (A_1, \ldots, A_{i-1}, A_{i+1}, \ldots, A_\maxAgent)$ denote the
action profile without $i$'s action. 

We refer to this stage of the process as the \emph{treatment selection
  game} in order to emphasize that agents (i.e., the treatments) can
be strategic in selecting the treatment version they administer to
units.  This differentiates our setting from traditional experimental
design, because it allows multiple versions of the same treatment to
be available, hidden to the experimenter, and subject to selection
by strategic agents.
The traditional setting of experimental design is recovered if all
action spaces of all agents are singletons, i.e., there is only one
treatment version for each agent.\footnote{Dealing with multiple hidden treatments 
remains 
  an open problem in 
traditional experimental design and causal inference,
  although not in a game theoretic setting as ours, and it is
  typically assumed away, for example, through SUTVA
  ~\cite{rubin1980comment}.}

\subsection{Outcomes}
\label{section:notation:outcomes}

Subsequent to the treatment selection game, an {\em outcome} is
measured on each experimental unit $u$.  Generally, the
\emph{potential outcome} of unit $u$, denoted by $Y_u(\Z, \A)$, is the
outcome that will be observed under assignment $\Z$ and agent actions
$\A$. We assume that outcomes are numerical values; e.g., expenditure in dollars, number of product purchases, etc.

However, only one potential outcome can be observed at any given
experiment, depending on the realized assignment $\Z$ and actions
$\A$, while the rest will be missing.  To emphasize the difference
between potential outcomes and \emph{observed outcomes}, we use
additional notation.  Let $\Yobs{ui}$ denote the observed outcome on
unit $u$ that was assigned to agent $i$. The notation $\Yobs{ui}$
implies that $u$ was assigned to $i$ (i.e., $Z_u=i$), and it is
undefined if $Z_u \ne i$, i.e., $u$ was not assigned to $i$. Following
a ``dot-notation," $\Yobs{.i}$ denotes the $\UperA \times 1$ vector of
observed outcomes of units assigned to agent $i$, and $\Yobs{..}$
denotes the $\maxUnit \times 1$ vector of observed outcomes of all
units.

Note the dependence of potential outcomes on the complete 
assignment vector $\Z$; this allows the outcome of unit $u$ 
to depend on assignment $Z_{u'}$ of some other unit $u'$, even when agent actions $\A$ are held fixed. 
This situation is reasonable, for example,
when units form social networks and influence each other, and is
generally known as \emph{social network interference}
Toulis and Kao \cite{toulis2013estimation}.  
In our setting, interference between units affects the actions agents take (treatment versions), which then affect the interference on units, and so on. We collectively refer to this situation as {\em strategic interference}.\footnote{
There exists work in experimental design 
with between-unit interference David and Kempton \cite{david1996designs}, although not under a strategic interference setting as ours. In this paper, 
we will not be concerned with such forms of interference, 
but it will be the focus of future work.
There is also related work in estimation of treatment effects 
in the context of strategic agents.
For example, Athey et. al. \cite{athey2008comparing} and Toulis and Parkes \cite{toulis2015statistical} evaluate mechanisms 
in terms of their revenue, under the causal framework of potential outcomes. In both papers, the treatments are two different mechanism formats, and the units are the agents 
competing in the mechanism.
The present work differs because, under our framework, the treatments are in fact strategic agents 
that are evaluated through an experiment, whereas the 
units passively exhibit treatment outcomes. See, also, 
the discussion by Dash \cite{dash2001caveats} on the challenges of causal inference in dynamical systems
within a different causal framework, namely causal graphs 
\cite{pearl2000causality}.}

We now illustrate the notation with an example application in viral
marketing, which we will reuse throughout this paper.

\newexample{Example 1}{ Assume four units $\units=\{1, 2, 3,4\}$ in a
  single block, say, undergraduate students, and two marketing agents
  $\agents = \{1, 2\}$. Further assume that 1 and 2 are close friends
and 3 and 4 are close friends.
The experimenter wants to understand which agent is
  better at advertising to  students. Assume a
  treatment assignment $\Z = (1, 2, 1, 2)^\intercal$, i.e., units
  $1,3$ are assigned to agent 1, and units $2,4$ to agent $2$. Each
  agent has two actions (treatment versions): advertise through phone
  or through social media. The action sets are thus $\actions_1 =
  \actions_2 = \{\text{phone, social}\}$, and a possible action
  profile is $\A = (\text{phone, social})^\intercal$ with $A_{1} =
  \text{phone}$ (agent 1 uses phone to reach units 1 and 3) and
  $A_{2}=\text{social}$ (agent 2 uses social media to reach units 2
  and 4.)

  The potential outcome $Y_u(\Z, \A)$ could denote the number of
  product purchases (integer outcome) made by unit $u$, or the net
  profit from advertising to unit $u$ (continuous outcome).
Dependence on the assignment and treatment
versions of both agents is reasonable because there could be
word-of-mouth effects between students.

Consider observed data $\Yobs{..}=(0, 1, 4, 1)^\intercal$; for
example, $\Yobs{31}=4$, which indicates that unit 3 was assigned to
agent 1 and purchased four product items; $\Yobs{32}$ is undefined
because the outcome of unit $3$ when assigned to agent 2 is not
observed. To illustrate the dot-notation, $\Yobs{.1} = (0,
4)^\intercal$ indicates the outcomes of units assigned to agent 1, and
$\Yobs{.2}=(1, 1)^\intercal$ indicates the outcomes for agent 2.}
\medskip

In Example~1, the experimenter might be tempted to declare agent 1 as
the winner, because it achieves $\Yobsbar{.1}=2.0$ purchases/unit, as
opposed to $\Yobsbar{.2}=1.0$ purchases/unit for agent 2. However,
these sample averages are subject to random variability from the
randomization in the experiment, and may result from actions that are not the natural actions of the agents.
%
%
Therefore, it is unclear whether the sample averages actually estimate
how agents would do if they were selecting treatments without
competition.

\newcommand{\scores}{\{\phi_b, b \in \m{B}\}}
\subsection{Estimand and estimators}
\label{sec:estimands}

A principled approach is to define the quantity of interest to the
experimenter, the {\em estimand}, and then devise appropriate
estimators for that quantity.  The estimand is the agent with best
possible performance, and thus we need a concrete notion of
performance.  For this, we want to estimate how good an agent's action
would be if it was played without competition and thus without
strategic interference.  This is important because, ultimately, the
experimenter wants to assign a contract (e.g., an advertising
campaign) to the winner agent, after which the winner will act
by itself.

Let's define the \emph{performance of agent $i$ with respect to its
  action $\alpha_i$}, denoted by $\chi(\alpha_i)$, as
\begin{align}
\label{eq:performance}
\chi(\alpha_i) = \ExCond{Y_u(\Z, \A)}{\A=\alpha_i \ones, Z_u=i};
\end{align}
notation $\A = \alpha_i \ones$ denotes the hypothetical 
situation where all agents other than agent $i$ are replaced by ``replicates" of $i$, and each replicate plays action $\alpha_i$.
The dependence of $\chi(\alpha_i)$ on agent index $i$ will be implicit
in the notation. Given assignment vector
$\Z$ and actions $\A$, we  assume that the distribution of
potential outcomes is known to all agents. 

The expectation in Eq.
\eqref{eq:performance} is taken with respect to this distribution,
and defines the quantity of interest to the experimenter 
 because it captures
how agent $i$ would do, on average, if the agent was acting alone 
without competition.\footnote{In causal
  inference, Eq. \eqref{eq:performance} is a \emph{superpopulation
  estimand}, where the experimental units are assumed to be a random sample
  from a superpopulation of units, which is the target of statistical
  inference. The expectation in Eq. \eqref{eq:performance} is thus
  over all units in the superpopulation and all treatment assignments,
  for fixed agent actions. Other estimands in that superpopulation are
  possible; for example, the experimenter might be interested in the
  median outcomes, $\mathrm{med}(Y_u(\Z, \A))$, or the Sharpe ratio,
  $\Ex{Y_u(\Z, \A)} / \mathrm{SD}(Y_u(\Z, \A))$, all conditional on
  fixed actions as in Eq. \eqref{eq:performance}.  In this paper, we
  work under the estimand of Eq. \eqref{eq:performance}, mainly for
  simplicity, however our theory applies to all aforementioned
  estimands as well.}  
We also refer to $\chi$ as the \emph{performance function}, and define
$\chiA = (\chi(A_1), \chi(A_2), \ldots, \chi(A_\maxAgent))^\intercal$.
For brevity, all following definitions for an agent, e.g.,
natural action, quality, etc., will be implicitly assumed to be stated
with respect to a particular performance function $\chi$.

The {\em natural action} of agent $i$ is
the action that maximizes the quantity of interest to the experimenter
in a system where agent $i$ acts alone without competition.
In particular, the \emph{natural  action} of agent $i$, denoted by $A_i^\star$, is defined as the action that maximizes its performance, i.e.,

\begin{align}
\label{eq:natural_action}
A_i^\star \defeq  \arg \max_{\alpha_i \in \actions_i} \left \{
\chi(\alpha_i)  \right \}.
\end{align}

The natural action profile is denoted by $\A^\star = (A_1^\star,
A_2^\star, \ldots, A_{\maxAgent}^\star)$.   
The \emph{quality} of agent $i$, denoted by $\chi_i^\star \in
\Reals$, is the maximum performance that the agent can
achieve, i.e., $\chi_i^\star = \chi(A_i^\star)$. 
The {\em estimand}, denoted by $\tau$, is the agent of
highest quality, i.e.,
\begin{align}
\label{eq:estimand}
\tau =  \arg \max_{i \in \agents} \{\chi_i^\star\}.
\end{align}

To estimate the agent of highest quality the experimenter needs to use
the observed outcomes $\Yobs{..}$.  We will assume that the
experimenter uses a \emph{score function} $\phi : \outcomes^\maxUnit
\to \Reals^\maxAgent$, mapping all outcomes to a $\maxAgent \times 1$
vector of scores for each agent, denoted by $\phi_i$ for agent $i$.
For convenience, we will write $\phiB = (\phi_1(\Yobs{..}),
\phi_2(\Yobs{..}), \ldots, \phi_\maxAgent(\Yobs{..}))^\intercal$.

In the experiment, agents will be evaluated according to their scores,
and the winner is the agent with the highest score. Several options
for the score functions are possible.  For example, $\phi_i(\Yobs{..})
= \Yobsbar{.i}$, the sample mean of outcomes of units assigned to
agent $i$, is one choice for the score function; other choices are
possible, e.g., the sample Sharpe ratio, the sample median, etc.

The key challenge in incentive-compatible experimental design is to
align maximizing the probability of winning the experiment, as induced in part by the score function $\phi$, with 
selecting the action with maximum performance, i.e., the natural action.

\subsection{Incentive-compatible experiment designs}
\label{section:notation:designs}

Let's first define an experiment design
using the concepts of estimand and estimators from Section \ref{sec:estimands}.
\begin{definition}
An \emph{experiment design} $\design = \designTuple$ operates 
in the following steps:
\begin{enumerate}
\item Receives units $\units$ and agents $\agents$, as input.
\item Samples a treatment assignment $\Z$ according to $\psi$.
\item Each agent $i$ picks a treatment version $A_i$, and 
	administers the treatment to the set of its assigned units, $\{u \in \units : Z_u=i\}$.
\item Outcomes on units $\Yobs{..}$ are observed.
\item The winner agent $\hat{\tau}$ is declared according to the rule
	\begin{align}
	\label{eq:winner}
	\winner = \arg \max_{i \in \agents} \left\{ \phi_i(\Yobs{..}) \right\}.
\end{align}
\end{enumerate}

\end{definition}
\medskip

Given experiment design $\design$ and action profile $\A$, the
probability $P_i(\A|\design)$ that agent $i$ wins the experiment is
given by:
\begin{align}
\label{eq:win_prob}
\mathrm{Pr}\left(\winner=i| \A, \design\right) \defeq P_i(\A | \design)
=P_i(\alpha_i, \A_{-i}|\design).
\end{align}

The randomness in Eq. \eqref{eq:win_prob} comes from the randomness of
observed data $\Yobs{..}$, and the randomization in the treatment
assignment.  The winning probability $P_i(\cdot|\design)$ in Eq.
\eqref{eq:win_prob} is the \emph{expected utility} of agent $i$ under action profile $\A$, because agents care only 
about winning the experiment.
%
%
\begin{definition}[Incentive-compatible experiment design] 
  An experiment design $\design=\designTuple$ is
  \emph{incentive-compatible} if the natural action $A_i^\star$ is a dominant strategy
  for each agent $i$, i.e., it maximizes the probability
  \eqref{eq:win_prob} of winning the experiment regardless of other
  agents' actions, such that 
\begin{align}
\label{def:ic}
\arg \max_{\alpha_i \in \actions_i} \{ P_i(\alpha_i, \A_{-i} | \design) \}= 
A_i^\star,
\end{align}
for all actions $\A_{-i}$, and every agent $i$.
\end{definition}

\emph{Remark.} 
In an incentive-compatible
experiment, the score function $\phi$ induces a probability of winning \eqref{eq:win_prob} that is monotonically
increasing with the performance function $\chi$ 
that the experimenter cares
about.  If this monotonicity holds, an agent will prefer to play the
action that maximizes its performance (i.e., the natural action),
because this will also maximize the winning probability.
\smallskip

The notation is summarized in Table~\ref{table:notation}.  
We now return to the viral marketing problem that was introduced 
in Example 1. Examples~2(a)-(c) deal with Normally-distributed outcomes,
whereas Examples~3(a)-(g) deal with Poisson-distributed outcomes.
Examples 3(c)-(g) deal specifically with the problem of interference, and work with a more realistic form of the viral marketing problem.


\renewcommand*{\arraystretch}{1.2}
\begin{table}[t!]
\centering
\small
\caption{Notation for incentive-compatible experimental design}{
\begin{tabular}{lll}
Symbol                  & Description                                                                                       & Value/Domain\\
\hline 
$\units$  & Set of $\maxUnit$ units & $ \{1, 2, \ldots, \maxUnit\}$ \\
$\agents$  & Set of $\maxAgent$ agents                                              & $\{1, 2, \ldots, \maxAgent \}$ \\
 $Z_u$ & Treatment assignment of unit $u$
& $Z_u \in \agents$ \\
$\Z$ & Vector of treatment assignment ($\maxUnit \times 1$) 
& $(Z_1, \ldots, Z_\maxUnit)^\intercal$ \\
$\UperA$                      & Units per agent   
& $\UperA = \maxUnit / \maxAgent$ \\
$\actions$         & Generic action space\\
$\actions_i$         & Action space of agent $i$
& $\actions_i \subseteq \actions$ \\
$A_i$                   & Action of agent $i$ 
& $A_i \in \actions_i$ \\
$\A$ & Complete action profile ($\maxAgent \times 1$)
& $(A_1, \ldots, A_\maxAgent)^\intercal$ \\
$Y_u(\Z, \A)$           & Potential outcome of unit $u$ under assignment $\Z$, actions $\A$
& $Y_u(\Z, \A) \in \Reals$ \\
$\Yobs{ui}$             & Observed outcome for unit $u$ assigned to agent $i$   \\
$\Yobs{.i}$             & Vector of observed outcomes of units assigned to agent $i$ ($\UperA \times 1$)
& $\Yobs{.i} \in \outcomes^{\UperA}$     \\
$\Yobs{..}$             & Vector of observed outcomes of all units
($\maxUnit \times 1$)
& $\Yobs{..} \in \outcomes^{\maxUnit}$ \\
$\chi(\alpha_i)$        & Performance of agent $i$ playing action $\alpha_i$
& $\chi(\alpha_i) \in \Reals$ \\
$\chiA$ & Vector of performances  ($\maxAgent \times 1$)  & $\left(\chi(A_1), \ldots, \chi(A_\maxAgent)\right)^\intercal$\\
$A_i^\star$             & Natural action of agent $i$ -- maximizes performance          
& $A_i^\star \in \actions_i$                 \\
$\chi_i^\star$          & Quality of agent -- performance at natural action
& $\chi_i^\star \in \Reals$ \\
$\tau$                  & Agent of highest quality                                            
& $\tau \in \agents$ \\
$\phi_i(\Yobs{..})$     & Score of agent $i$                                                                                   
& $\phi_i(\Yobs{..}) \in \Reals$  \\
$\phiB$ & Vector of agent scores ($\maxAgent \times 1$) & $\left(\phi_1(\Yobs{..}),  \ldots, \phi_\maxAgent(\Yobs{..})\right)^\intercal$ \\
$\hat{\tau}(\Yobs{..})$ & Estimated agent of highest quality -- agent with maximum score
& $\hat{\tau}(\Yobs{..})  \in \agents$                               \\
$P_i(\A | \design)$     &                                                                            Probability agent $i$ wins under design $\design$, given fixed actions 
$\A$
\normalsize
\end{tabular}}
\label{table:notation}
\end{table}
\normalsize

\newexample{Example 2(a). -- Normal outcomes\footnote{This two-agent example (low-quality agent vs. high-quality agent) is different from the example in the original paper published at EC'2015. The example was edited to illustrate a scenario where the low-quality agent prefers to play 
an action that is not its natural action and also reduces 
the winning chances of the high-quality agent. 
In the example of the original paper, the deviation 
from the low-quality agent actually increased the chances of the high-quality agent.}}{ 
Consider the viral
  marketing problem of Example 1, with multiple units 
and two agents, where the outcomes
  of interest are the profit achieved 
from advertising to each unit.
  We  assume that an agent action $\alpha_i= (\mu_i, \sigma_i^2) \in
  \Reals{} \times \Reals{}^+$, determines the mean and variance of the
  profit from advertising to unit $u$, such that, given assignment 
$\Z$, actions $\A$,
\begin{align}
\label{eq:normal}
Y_u(\Z, \A)  \sim \mathcal{N}(\mu_i, \sigma_i^2), \text{ if } 
 A_i=\alpha_i, Z_u=i.
\end{align}
The profit defined in Eq. \eqref{eq:normal} 
can be negative because we assume implicit 
advertisement costs. Furthermore Eq. \eqref{eq:normal} 
implies  no interference between units, and no strategic interference 
between agent actions. We will make this precise in Section~\ref{section:theory}.

The experimenter is interested only in expected profit, ignoring the risk.
%
Thus, the performance of action $\alpha_i=(\mu_i, \sigma_i^2)$ of agent $i$ is
\begin{align}
\label{eq:performance_normal}
\chi(\alpha_i) \defeq \ExCond{Y_u(\Z, \A)}{\A=\alpha_i \ones, Z_u=i} = \mu_i.
\end{align}

Hence, the quality $\chi_i^\star$ of agent $i$ is the maximum $\mu_i$
the agent can achieve over its action space $\actions_i$.
Now, consider an experiment design $\design=\designTuple$, where the
score function $\phi$ is defined as $\phi_i(\Yobs{..}) =
\Yobsbar{.i}$, i.e., the score of agent $i$ is the sample mean profit
from all units assigned to agent $i$.  Ignoring ties, the winning agent is given
using Eq. \eqref{eq:winner}:
	\begin{align}
	\label{eq:winner:normal}
	\winner =\begin{cases} 
		1, &\mbox{if }\Yobsbar{.1} > \Yobsbar{.2}, \\ 
		2, &\mbox{if }\Yobsbar{.1} < \Yobsbar{.2}. 
\end{cases}
\end{align}

By Eq. \eqref{eq:normal}, $\Yobsbar{.i} \sim \mathcal{N}(\mu_i, \sigma_i^2/ \UperA)$, where $k$ is the 
number of units per agent. Hence, the probability that agent 1 wins is
\begin{align}
\label{eq:win_one}
P_1(\A|\design) \defeq \mathrm{Pr}\left(\winner=1| \A, \design\right) =
P(\Yobsbar{.1} > \Yobsbar{.2}) = 
\Phi(\sqrt{\UperA} \frac{\mu_1-\mu_2}{\sqrt{\sigma_1^2 + \sigma_2^2 }}),
\end{align}
where $\Phi$ is the normal cumulative distribution function (CDF).
This design is not incentive-compatible because the winning
probability $P_1(\A|\design)$ is not monotone with
performance $\chi(\alpha_1) = \mu_1$ for action $\alpha_1 = (\mu_1,
\sigma_1^2)$.  For example, an increase in $\mu_1$ may be
associated with an increase in the risk $\sigma_{1}^2$, such that the
 probability of winning is reduced.

 To see this, assume there are only two actions for agent 1, which
 induce mean and variance $\actions_1 = \{(1.5, 100), (2, 20)\}$,
and only one action for agent 2, $ \actions_2 = \{(9, 1)\}$.
 The quality of agent 1 is $\chi_1^\star \defeq \max \{\mu : (\mu, \sigma^2) \in \actions_1\} = 2$ and thus $(2, 20)$ is agent 1's natural action.  
However, when agent 1 plays the natural action, its winning probability is approximately equal to 0.12, whereas action $(1.5, 100)$ yields winnining probability 0.364, approximately.
When agent 1 does not play the natural action, the 
expected value of its outcomes are reduced but 
their variance is increased, thus overall increasing agent 1's chances to win the experiment.
%
Therefore, this experiment is not incentive compatible 
since agent 1 prefers not to play the natural action.
}

\newexample{Example 2(b). -- Normal outcomes -- High
  risk/reward}{
Continuing Example 2(a), let's suppose that
  the variance of the unit's outcome satisfies $\sigma_i^2 = \mu_i^4$, 
indicating a delicate 
trade-off between expected return and
  risk.  The probability that agent 1 wins is easily obtained from
  \eqref{eq:win_one} as,
\begin{align}
\label{eq:win_one:curved}
P_1(\A|\design) = P(\Yobsbar{.1} > \Yobsbar{.2}) = 
\Phi(\sqrt{\UperA} \frac{\mu_1-\mu_2}{\sqrt{\mu_1^4 + \mu_2^4 }}).
\end{align}
The experiment design is still not incentive-compatible because
\eqref{eq:win_one:curved} is not increasing monotonically
with $\mu_1$.  As before, the better agent will
choose to be more conservative, and will not reveal its quality (maximum
possible $\mu_1$).  However, we will show in
Section~\ref{section:theory} that an incentive-compatible design can
be achieved through the score function $\phi_i(\Yobs{..}) = -
1/\Yobsbar{.i}$, i.e., the negative reciprocal of the sample mean profit. We will show that, with this score function, the risk-reward trade-off in \eqref{eq:win_one:curved} disappears, 
which allows the experimenter to estimate agents' qualities.  }

\newexample{Example 3(a) -- Poisson outcomes}{
Now suppose the outcomes 
  are integer-valued, e.g., 
 representing the number of purchases. In this case, 
we assume that an 
  agent's action $\alpha_i = (\lambda_i) \in \Reals{}^+$ determines the
  purchase rate by unit $u$, such that, given assignment 
$\Z$, actions $\A$,
\begin{align}
\label{eq:poisson}
Y_u(\Z, \A)  \sim \mathrm{Pois}(\lambda_i), \text{ if } 
 A_i=\alpha_i, Z_u=i.
\end{align}

As in Eq.~\eqref{eq:normal} of Example~2(a), Eq. \eqref{eq:poisson} implies no interference.
Let's suppose
the experimenter is interested in
performance that is the 
expected purchase rate.
Thus, using Eq. \eqref{eq:performance}, the experimenter measures performance of action $\alpha_i=(\lambda_i)$ of agent $i$, through 
\begin{align}
\label{eq:performance_poisson}
\chi(\alpha_i) \defeq \ExCond{Y_u(\Z, \A)}{\A=\alpha_i \ones, Z_u=i} = \lambda_i.
\end{align}
Hence, the quality $\chi_i^\star$ of agent $i$ is the maximum purchase rate $\lambda_i$ that the agent can achieve 
over its action space $\actions_i$.
Now, consider the experiment design $\design=\designTuple$, where the score function $\phi$ is defined as $\phi_i(\Yobs{..}) = \Yobsbar{.i}$, i.e., the score of agent $i$ is the sample mean purchase rate from all units assigned to agent $i$.  
Ignoring ties, the
winning agent $\hat{\tau}(\Yobs{..})$ is given using Eq. \eqref{eq:winner:normal}.
By the central limit theorem, $\Yobsbar{.i} \tod \mathcal{N}(\lambda_i, \lambda_i/\UperA)$, where ``$\tod$" denotes 
convergence in distribution, and $\UperA$ 
is the number of units per agent.
The probability that agent 1 wins is,
asymptotically,
\begin{align}
\label{eq:win_one:poisson}
P_1(\A|\design) =
P(\Yobsbar{.1} > \Yobsbar{.2}) = 
\Phi(\sqrt{\UperA} \frac{\lambda_1-\lambda_2}{\sqrt{\lambda_1 + \lambda_2} }).
\end{align}

This design is incentive-compatible because 
the winning probability $P_1(\A|\design)$ is  monotone
with the agent performance; for example, an increase in $\lambda_1$ incurs a larger increase in the nominator of Eq. \eqref{eq:win_one:poisson} than in the denominator. By symmetry, the winning probability for agent $i$ is maximized at its natural action.
\smallskip

In Section~\ref{section:no_interference:most_powerful}, we will show
that a more powerful design is possible, i.e., there exists an
experiment design $\design'$ that is incentive-compatible and also
guarantees higher winning chances to the better agent.  }

The examples highlight the challenges in incentive-compatible
experimental design that arise because the experimenter is interested
in some quality of an agent (e.g., expected return) but cannot find a
design that incentivizes agents to play in a way that reveals their
qualities.  The problem that can arise is because of a mismatch
between the score function $\phi$ that is used to declare the winner,
and its effect in inducing a non-cooperative game, and the performance
function $\chi$ that is of interest to the experimenter.

Compared with classical {\em mechanism design theory}, incentive-compatible
experimental design differs in that:
\begin{itemize}
\item In mechanism design, the private information is an agent's
  preferences, whereas here the private information is an agent's
  quality (i.e., the performance of its natural action).
\item In mechanism design, there may be side payments that can be
  made, whereas here the incentives are winner-take-all and depend on
  the outcome of the experiment.
\item In mechanism design, it is standard to appeal to the {\em
    revelation principle} and design a direct-revelation mechanism, in
  which agents report their preference type to the mechanism. In
  comparison, the agents in our setting select an action and the
  designer observes the effect of this action, but not the action
  itself.
\end{itemize}

\newcommand{\Jacob}{\mathcal{J}_{\phi}}
\section{Theory of incentive-compatible experimental design}
\label{section:theory}

In this section we prove our main result, which provides a
construction of score functions to design incentive-compatible
experiments.  The proof relies on the existence of statistics that can
estimate the individual agent performances $\chi(A_i)$, as the number of units grows large.

\begin{definition}[Identifiable performance, identifying statistic]
\label{def:identification}
An experiment design $\design=\designTuple$ has {\em identifiable performance} $\chi$, if for every fixed action profile $\A$, there exists a statistic $T : \outcomes^\maxUnit \to \Reals^{\maxAgent}$ calculated over data $\Yobs{..}$, such that
\begin{align}
\label{eq:identifying}
\sqrt{\UperA} \left(T(\Yobs{..}) - \chiA\right) \tod \N{\Sa},
\end{align}
as the number of units per agent $\UperA$ grows large; $\mathcal{N}$ is the $n$-variate standard normal, and $\Sa$ is 
the $\maxAgent \times \maxAgent$ covariance matrix of $T$
that can depend on $\A$. The statistic $T$ is 
an {\em identifying statistic} for experiment design $\design$.
\end{definition}

An identifying statistic is important because it
estimates the individual performances $\chi(A_i)$, which are the
quantities of interest to the experimenter. Although finding such a
statistic is not an easy task, one simple strategy is to use sample
quantities, such as averages, and then appeal to the central limit
theorem, or other large-sample asymptotic results. We use this
strategy extensively in this paper. 

However, an identifying statistic $T$ calculated over data $\Yobs{..}$
need not be sufficient for incentive alignment in our winner-take-all
experiments. Thus, we consider score functions defined as
$\phi_i(\Yobs{..}) = f(T_i)$, for an appropriate transformation $f :
\Reals \to \Reals$. The transformation is used to add flexibility in
the design of the score function.  Agents will be evaluated according
to the score vector $\phiB$. The covariance matrix of the score vector
$\phiB$ is, asymptotically, equal to
\begin{align}
\label{eq:Va}
V_f(\A) = \Jacob \Sa \Jacob^\intercal,
\end{align}
where $\Jacob$ is the 
Jacobian of $\phi$ calculated at $\chiA$, 
actually a diagonal matrix with elements $f'(\chi(A_i))$.
Whether an experiment design $\designTuple$ is incentive-compatible or not, 
depends crucially on the matrix $V_f(\A)$ because this matrix 
defines the variances of the scores used to evaluate 
the agents. 
\begin{theorem}
\label{theorem:general}
Fix agent actions $\A$, and consider design $\design=\designTuple$ that has an 
identifying statistic $T$ with covariance matrix $\Sa$. 
Define the score function as $\phi_i(\Yobs{..}) =f(T_i)$, for some function $f : \Reals \to 
\Reals$, and let $v_{ij}(\A)$ be the $ij$th element 
of $V(\A)$ defined in Eq. \eqref{eq:Va}.
Also define,
\begin{align}
\label{eq:theorem:cov}
v^{ij}_f(\alpha|\A_{-i}) = v_{ii}(\alpha, \A_{-i}) + v_{jj}(\alpha, \A_{-i}) - 
v_{ij}(\alpha, \A_{-i}) - v_{ji}(\alpha, \A_{-i}).
\end{align}
Design $\design$ is incentive-compatible, if, for every agent $i$,
\begin{align}
\label{eq:theorem:general}
\arg \max_{\alpha_i \in \actions_i}   \left\{ \frac{f(\chi(\alpha_i))}
{ v^{ij}_f(\alpha_i|\A_{-i})^{1/2}} \right\}  & = \arg \max_{\alpha_i \in \actions_i}   \left\{ \chi(\alpha_i)\right\} \defeq A_i^\star,
\end{align}
for every agent $j\neq i$, and all actions $\A_{-i}$.
\end{theorem}


For a fixed action profile $\A$, the element $v^{ij}_f$ in Eq.
\eqref{eq:theorem:general}, is the variance of the difference between
the scores of agents $i$ and $j$, $\phi_i(\Yobs{..}) -
\phi_j(\Yobs{..})$, as defined in Theorem \ref{theorem:general}.
Thus, Eq. \eqref{eq:theorem:general} is the probability that agent $i$
has a larger score than agent $j$, and implies that this probability
is maximized at the natural action.

Theorem \ref{theorem:general} suggests a recipe to construct
incentive-compatible experiments, as we illustrate through examples in
the following sections.
\begin{itemize}
\item First, one needs to find an identifying
statistic to estimate the performances of agents, i.e., their outcomes
without competition.  A parametric model for the unit outcomes
together with known asymptotic results, such as the central limit
theorem, or the asymptotic normality of the maximum-likelihood
estimator, can provide such an identifying statistic with known
covariance matrix $\Sa$; see also Appendix \ref{section:discussion} 
for a relevant discussion.
\item Second, given the identifying statistic, one then needs to find 
an appropriate transformation $f$ to satisfy Eq. \eqref{eq:theorem:general}. This transformation can be as simple 
as the identity function, as in Example 3(g),
or the reciprocal function, as in Example 2(c).
Intuitively, the design goal for $f$ is to make the denominator of \eqref{eq:theorem:general} less sensitive to agent actions than the nominator.
\end{itemize}

Theorem~\ref{theorem:general} makes no assumption about
interference. In the following sections, we will specialize and apply
Theorem~\ref{theorem:general} on the viral marketing example, both
with and without interference.


\section{Incentive-compatible experiments without interference}

The setting without interference is formally defined 
through the following assumption.
%
%
\begin{assumption}[No interference]
\label{assumption:no_interference}
There is no strategic interference among agents 
and no interference between units, i.e., for all
assignments $\Z$ and all agent actions $\A$,
\begin{align}
\label{eq:sutva}
& Y_u(\Z, \A) \equiv Y_u(A_i), \doublespace \text{where } 
Z_u = i.
\end{align}
\end{assumption}

Assumption \ref{assumption:no_interference} postulates that the
potential outcome $Y_u(\Z, \A)$ of a unit $u$ assigned to agent $i$,
remains constant as long as agent $i$'s action and 
unit $u$'s assignment to agent $i$ are held fixed.
Under no interference, the distribution of a score function 
defined through an identifying statistic is a 
univariate normal, as shown in the following proposition.
\begin{proposition}
\label{proposition:no_interference}
Consider design $\design = \designTuple$ with an identifying statistic $T$ with covariance matrix $\Sa$.
Let $\phi_i(\Yobs{..})= f(T_i)$, for some function $f : \Reals \to \Reals$, and suppose Assumption \ref{assumption:no_interference} holds. Then, 
for fixed actions $\A$,
\begin{align}
\label{eq:phi}
\sqrt{k} \left(\phi_i(\Yobs{..}) - f(\chi(A_i)\right) \tod \mathcal{N}(0, \sigma^2(A_i)),
\end{align}
where $\sigma^2(A_i) = f'(\chi(A_i))^2 \sigma_{ii}^2$, with 
$\sigma_{ii}^2$ being the $i$th diagonal element of $\Sa$. 
\end{proposition}
\begin{proof}
By Assumption \ref{assumption:no_interference} (no interference),
the covariance matrix $\Sa$ of $T$ is diagonal 
with elements $\sigma_{ii}^2$. Thus, by  
definition of the identifying statistic,
\begin{align}
\sqrt{k} (T_i-\chi(A_i)) \tod \mathcal{N}(0, \sigma_{ii}^2). \nonumber
\end{align}
Since, $\phi_i(\Yobs{..}) = f(T_i)$, Eq. \eqref{eq:phi} 
follows from a simple application of the Delta theorem;
see, for example,
Bickel and Docksum \cite[Chapter 5]{bickel2001mathematical}, or Cox \cite{cox1998delta}.
\end{proof}

Proposition \ref{proposition:no_interference} provides 
the asymptotic distribution of the score function, given 
an identifying statistic and a known transformation $f$, 
when there is no interference. This will be useful to 
derive the winning probabilities for agents in the experiment.
We first illustrate Proposition
\ref{proposition:no_interference}, and 
then show how it can be used to simplify
the conditions of the more general Theorem \ref{theorem:general}.

\newexample{Example 2(c)}{
We continue from Example 2(b), 
where agent $i$'s action is $A_i=(\mu_i)$,
and  $\Yobsbar{.i} \sim \mathcal{N}(\mu_i,
  \mu_i^4/\UperA)$, where $\UperA$ is the 
number of units per agent.
The statistic 
$T(\Yobs{..}) = (\Yobsbar{.1}, \Yobsbar{.2}, \ldots, \Yobsbar{.\maxAgent})^\intercal \equiv T$, is an identifying statistic, since 
$\chiA = (\mu_1, \mu_2, \ldots, \mu_{\maxAgent})^\intercal \defeq \b{\mu}$, and
\begin{align}
\sqrt{\UperA} (T - \b{\mu}) \tod \mathcal{N}(0, \Sigma),
\end{align}
where $\Sigma = \mathrm{diag}(\mu_1^4, \ldots, \mu_\maxAgent^4)$, is the diagonal matrix with elements $\mu_i^4$.

Consider the score functions $\phi_i(\Yobs{..}) =
  1/T_i =1/\Yobsbar{.i}$, i.e., 
$f(x)=1/x$, in the notation of Proposition 
\ref{proposition:no_interference}. Using the result in Proposition
\ref{proposition:no_interference}, $\sigma^2(A_i) = f'(\mu_i)^2 \mu_i^4 = 1$, and thus
\begin{align}
\label{eq:phi_const}
\sqrt{\UperA}(\phi_i(\Yobs{..}) - 1/\mu_i) \tod \mathcal{N}(0,1).
\end{align}
}
The variance of the score function in Eq. \eqref{eq:phi_const} 
is stabilized. The following theorem shows 
that such variance stabilization can lead to incentive-compatible 
designs, when there is no interference.
%
%
\begin{theorem}
\label{theorem:no_interference}
Consider design $\design = \designTuple$ with an identifying statistic $T$ with covariance matrix $\Sa$.
Suppose Assumption \ref{assumption:no_interference} holds. If, for every agent $i$,
\begin{align}
\label{eq:condition:simple}
 & \phi_i(\Yobs{..}) = f(T_i), \text{ where } f:\Reals\to\Reals, \\ 
\label{eq:condition:const}
 & \Var{\phi_i(\Yobs{..})}  = \mathrm{const.},\\
\label{eq:condition:monotone}
 & \arg \max_{\alpha_i \in \actions_i} f(\chi(\alpha_i)) = \arg \max_{\alpha_i \in \actions_i} \{\chi(\alpha_i)\} \defeq A_i^\star,
\end{align}
then design $\design$ is incentive-compatible.
\end{theorem}

Condition \eqref{eq:condition:const} is related to
variance-stabilizing transformations in statistics, which also play
an important role in hypothesis testing; we discuss this relationship
in Appendix \ref{section:varstab}.

\newexample{Example 2(d). -- Normal outcomes -- High risk/reward} {
Continuing from Example 2(c), we consider the high risk-reward setting of the viral marketing problem,
where an agent's action is to pick an expected return, i.e., 
$A_i = (\mu_i)$, and the winning probability is given by
\begin{align}
\label{eq:curved}
P_1(\A|\design)= \Phi(\sqrt{\UperA} \frac{\mu_1 - \mu_2}{\sqrt{\mu_1^4 + \mu_2^4}}).
\end{align}

The performance function is $\chi(\alpha_i) = \mu_i$, 
and thus the natural action is $A_i^\star = \arg \max_{\alpha_i \in \actions_i} \{ \alpha_i \}$.
It was shown that design $\design$ in Example 2(b) --using 
the sample mean as the score function-- 
is not incentive-compatible. Consider 
instead a design $\design'$ with score function
$\phi_i(\Yobs{..}) = -1/\Yobsbar{.i}$. 
Using the result of Example 2(c), 
\begin{align}
\label{eq1}
\sqrt{\UperA} \left(\phi_i(\Yobs{..}) - (-1/\mu_i)\right) \tod \mathcal{N}(0, 1).
\end{align}

Condition \eqref{eq:condition:simple} is satisfied by definition 
of $\phi_i$. Condition 
\eqref{eq:condition:const} is also satisfied, because 
the variance of $\phi_i(\Yobs{..})$ in Eq. \eqref{eq1} is constant. Furthermore, 
\begin{align}
\arg \max_{\alpha_i \in \actions_i}   \left\{ f(\chi(\alpha_i)) \right\}
= \arg \max_{\alpha_i \in \actions_i}   \left\{ -1/\alpha_i \right\} =  \arg \max_{\alpha_i \in \actions_i} \{ \alpha_i \} = A_i^\star \nonumber,
\end{align}
which satisfies Condition \eqref{eq:condition:monotone}.
Thus, all conditions of Theorem \eqref{theorem:no_interference} are 
fulfilled. It follows that the new design $\design'$ is incentive-compatible.}

By construction of the probabilistic
model in Example 2(b), there is a very delicate trade-off
between expected return (agent performance) and risk; for example, if an
agent doubles its performance, then the risk will quadruple. In such
situations, it is a bad idea to adopt the sample mean as the score
statistic.  Intuitively, Eq. \eqref{eq:curved} shows that the
higher-quality agent will try more conservative actions, thus hiding
its true quality. However, if agents are scored according to the
negated reciprocal of their sample mean, the probability that an agent wins
increases monotonically with an agent's performance. 
Thus, agents have
the incentive to select actions that maximize their performance, and thus it is a dominant
strategy to select their natural action. 

\subsection{Powerful incentive-compatible experiment designs}
\label{section:no_interference:most_powerful}

Given the choice of two incentive-compatible designs, it is natural to
prefer the design in which the highest-quality agent has the highest
probability of winning. We formalize this intuition through the
following definition.
%
\begin{definition}[Powerful incentive-compatible design]
\label{definition:powerful}
Consider two experiment designs $\design$ and $\design'$ that are
both incentive-compatible and operate on the same set of units
$\units$. Let $\tau$ be the agent of highest quality.  Design
$\design'$ is (weakly) \emph{more powerful} than design $\design$ if the
probability that agent $\tau$ wins 
in the dominant strategy equilibrium
is higher in $\design'$ than
$\design$;
i.e.,
\begin{align}
\label{eq:powerful}
 P_{\tau}(\A^\star | \design')  \ge  P_{\tau}(\A^\star | \design),
\end{align}
where $\A^\star$ is the natural action profile, which is the same in both 
designs.
\end{definition}

In the following theorem, we give a simple case 
where we can
transform an incentive-compatible design into a more powerful one. 
\begin{theorem}
\label{theorem:rao}
Consider an incentive-compatible design $\design = \designTuple$, 
where action sets $\actions_i \subseteq \Reals$ are compact, 
and performance $\chi$ is one-to-one and continuous.
Let,
\begin{align}
\label{eq:rao}
\sqrt{\UperA} \left(\phi_i(\Yobs{..})-\chi(A_i)\right) \tod
\mathcal{N}(0, \sigma^2(A_i)),
\end{align}
where function $\sigma^2 : \actions \to \Reals^{+}$ satisfies
\begin{align}
\label{eq:powerful:var}
\chi(\alpha_i') \ge \chi(\alpha_i) \Rightarrow \sigma^2(\alpha_i') \ge \sigma^2(\alpha_i),
\end{align}
for every agent $i$, and all actions $\alpha_i', \alpha_i \in \actions_i$.\footnote{Condition~\eqref{eq:powerful:var} posits 
that an agent cannot increase its expected 
score without increasing the variance of the score. This is a reasonable assumption in practice because actions that 
do increase the expected score without increasing the variance, 
are strongly preferred.}

Consider a design $\design' = (\psi, \phi')$,
where $\phi'_i(\Yobs{..}) = \nu(\phi_i(\Yobs{..}))$, for each agent $i$, with $\nu(\cdot)$ defined by
\begin{align}
\label{eq:rao:nu}
\nu(y) = \int^y \frac{1}{\sqrt{\sigma(\chi^{-1}(z)})} dz.
\end{align}
Then, design $\design'$ is incentive-compatible and more powerful than $\design$, if $\nu(\cdot)$ is convex,
or $1/\sqrt{\sigma^2(\chi^{-1}(\cdot))}$ and 
$\sigma^2(\chi^{-1}(\cdot))$ are both convex.
\end{theorem}

The variance of the new score function,  $\Var{\phi_i'(\Yobs{..})}$, is 
constant, because function $\nu$ defined in Eq. \eqref{eq:rao} is a
variance-stabilizing transformation \cite{cox1998delta}.
This fulfills Condition \eqref{eq:condition:const} of Theorem 
\ref{theorem:no_interference}, while the monotonicity \eqref{eq:powerful:var} of $\sigma(\cdot)$
maintains the monotonicity Condition \eqref{eq:condition:monotone}.
The new design 
$\design'$ is thus incentive-compatible.

%
\newexample{Example 3(b) -- Poisson outcomes}{
Continuing from Example 3(a), the actions are $A_i = (\lambda_i) \in \Reals^{+}$ 
with performance $\chi(A_i) = \lambda_i$, 
while the score statistic is $\phi_i(\Yobs{..}) =
  \Yobsbar{.i}$; thus, $\sqrt{\UperA} \left(\phi_i(\Yobs{..})-\lambda_i \right) \tod \mathcal{N}(0, \lambda_i)$. Let agent 1 be the best agent. Consider a new design $\design'$ with the transformation 
\begin{align}
\nu(y) = \int^y \frac{1}{\sqrt{\sigma(\chi^{-1}(z)})} dz = 
 \int^y \frac{1}{\sqrt{z}} dz = 2 \sqrt{z},\nonumber
\end{align}
and score function $\phi_i'(\Yobs{..}) = 
\nu(\phi_i(\Yobs{..})) = 2 \sqrt{\Yobsbar{.i}}$. 
Design $\design'$ is incentive-compatible and more powerful 
than design $\design$ of Example 3(a) by Theorem \ref{theorem:rao}, since
$1/\sqrt{\sigma^2(\chi^{-1}(z))} = 1/\sqrt{z}$ and 
$\sigma^2(\chi^{-1}(z)) = z$, are both convex.
Another way to see this is through Proposition \ref{proposition:no_interference}, which implies
$\sqrt{k}\left(\phi'_i(\Yobs{..}) - 2\sqrt{\lambda_i}\right) \tod \mathcal{N}(0, 1)$.
Thus, the probability that agent 1 wins is
\begin{align}
\label{eq:poisson_win_powerful}
P_{1}(\A|\design')= \Phi(\sqrt{2\UperA}(\sqrt{\lambda_{1}}- \sqrt{\lambda_{2}})).
\end{align}

We can verify $P_{1}(\A|\design') > P_{1}(\A|\design)$ by comparing 
Eq. \eqref{eq:poisson_win_powerful} with Eq. \eqref{eq:win_one:poisson}:
\begin{align}
 \Phi\left(\sqrt{2\UperA}(\sqrt{\lambda_{1}}- \sqrt{\lambda_{2}})\right) & >  \Phi\left(\sqrt{\UperA}
	 \frac{\lambda_{1} - \lambda_{2}}{\sqrt{\lambda_{1} + \lambda_{2}}}\right)   \Leftrightarrow  \sqrt{2} (\sqrt{\lambda_{1}}- \sqrt{\lambda_{2}}) > \frac{\lambda_{1} - \lambda_{2}}{\sqrt{\lambda_{1} + \lambda_{2}}} \nonumber.
\end{align}
The last inequality always holds because it reduces to $(\sqrt{\lambda_{1}}- \sqrt{\lambda_{2}})^2 > 0$.
}

In Example 3(b), the better agent (agent 1) has higher chances of winning in the new design $\design'$.
Since $\design'$ is also incentive-compatible, it follows that 
$\design'$ is more powerful than $\design$.
Intuitively, the square root transformation in the new design 
stabilizes the variance -- there is no denominator 
in Eq. \eqref{eq:poisson_win_powerful} -- which 
achieves incentive-compatibility through Theorem \ref{theorem:no_interference}. 

\subsection{Using transformations for more powerful designs}

If there is only one block and transformation $\nu(\cdot)$ 
in Theorem \ref{theorem:rao} is order-preserving,
then the transformation might not affect the power 
of the experiment design.
For a simple argument, let $X, Y$ be two positive random variables,
then $P(X>Y) = P(\nu(X) > \nu(Y))$ if $\nu$ is
order-preserving.\footnote{A similar observation can be made in
  regard to the use of score functions $\phi_i$ to achieve incentive
  compatibility: order-preserving transformations $\phi_i$ do not
  affect incentives. Note, for example, that the negated reciprocal
  transformation that aligns incentives in Example 2(d) is not
  order-preserving (e.g., $2 > -1$ but $-1/2 < -1/(-1)$). The outcomes
  in that example could take negative values; if outcomes were
  constrained to be positive, incentives would not be affected.}

However, when there are multiple blocks, 
a transformation can improve the power of the design
even when the transformation 
is order-preserving.
In the following simulation study, 
we expand the design introduced in Example 3(a) to 
multiple blocks in order to illustrate the positive effect of the square-root 
transformation, which is variance-stabilizing for 
Poisson outcomes, on the power of the design.
In this simulation study we focus on power because 
the design is already incentive-compatible, as shown 
in Example 3(a).\footnote{The introduction 
of multiple blocks does not affect the incentives because 
incentive compatibility was defined with respect 
dominant-strategy equilibrium and outcomes are sampled 
independently across blocks. Multiple blocks could affect 
incentives if agents were able to benefit from making strategic 
trade-offs between blocks, e.g., be conservative in one block 
and be risky in another.}

Consider a design $\design$ with two agents and two blocks. 
Agent $i$ plays action $\lambda_{ib}$ in block $b$; 
we set $\lambda_{11}=5, \lambda_{12}=10$ for agent 1, 
and $\lambda_{21}=4.25, \lambda_{22}=9.95$, and thus 
agent 1 is the high-quality agent.
We repeat the following process $10,000$ times.
First we fix the number of units per block, say $k$.
Second, we sample $Y_{uib} \sim \mathrm{Pois}(\lambda_{ib})$ 
i.i.d. for every unit $u$
in block $b$, where $Y_{uib}$ indicates the total number of sales for unit $u$ 
assigned to agent $i$ in block $b$.
We then use the sample mean as the default score function, 
but also apply a transformation $\nu$. In particular, the 
total score of agent $i$ is $\sum_{b=1}^2 \nu(\overline{Y_{. ib}})$, 
where $Y_{. ib} = (Y_{uib})$ is the vector of unit outcomes 
for agent $i$ in block $b$, and $\nu$ is the transformation.
The winner is the agent with highest score.
After all $10,000$ repetitions we report the \%wins by agent 1.

The results are shown in Table \ref{table:multiblock} where we compare
the identity transformation against the square-root transformation 
for multiple number of units per block.
We observe that the square-root transformation, which is also 
the variance stabilizing transformation according to Theorem \ref{theorem:rao}, 
increases the winning chances of agent 1 (high-quality agent).
As the number of units per block increases the sample means 
get closer to the actions played by the agents (i.e., values $\lambda_{ib}$) 
and thus agent 1 wins almost with probability one at 
both designs.

For an intuition why variance stabilizing works with multiple 
blocks, consider the argument at the beginning of this section.
In particular, let $X_1\defeq \overline{Y}_{.11}, X_2 \defeq \overline{Y}_{.12}$ be the sample means of agent 1 in blocks 1 and 2, respectively, and let $Y_1 \defeq \overline{Y}_{.21}, Y_2\defeq\overline{Y}_{.22}$, be the respective sample means
for agent 2. 
If there was no transformation the winning probability for agent 1 would be $P(X_1 + X_2 > Y_1 + Y_2)$. 
With the square-root transformation this 
probability is $P(\sqrt{X1}+\sqrt{X_2} > \sqrt{Y_1}+\sqrt{Y_2})$,
which is generally larger than the probability without transformation. 
Intuitively, the square-root transformation accentuates 
the differences in the mean-rates of the two agents (i.e., the actions 
$\lambda_{ib}$) and downplays the differences in the tails.
The formal proof is a simple extension of Theorem \ref{theorem:rao}, which uses convexity/concavity arguments.

\begin{table}[ht]
\centering
\caption{Probability agent 1 wins in a design with two blocks and two possible 
score transformations. Probabilities were calculated over 10,000 repetitions.}
\begin{tabular}{lcc}
  \hline
 &\multicolumn{2}{c}{Transformation  $\nu$} \\
\#units/block & $\nu(x)=x$ & $\nu(x) = \sqrt{x}$ \\ 
  \hline
5 & 0.62 & 0.65 \\ 
  10 & 0.67 & 0.72 \\ 
  25 & 0.77 & 0.82 \\ 
  50 & 0.85 & 0.91 \\ 
  100 & 0.93 & 0.97 \\ 
  500 & 1.00 & 1.00 \\ 
  1000 & 1.00 & 1.00 \\ 
   \hline
\end{tabular}
\label{table:multiblock}
\end{table}

\section{Incentive-compatible experiments with interference}
\label{section:interference}

We now consider strategic interference, whereby an action of an agent
can affect the outcomes of units assigned to another agent.
Therefore, agent scores calculated on individual agent outcomes are
confounded with the entire action profile.

\newexample{Example 3(c) -- Poisson outcomes with interference}{
  Building upon Example~3(b), we now introduce a more realistic model of the viral marketing experiment, which we assume operates as  follows.  

As before, units are assigned to agent 1 or agent 2.
  We refer to the units assigned to agent $i$, i.e., 
the set $\{u \in \units : Z_u=i\}$, as the \emph{test set}
  of agent $i$. In addition, each agent is free to pick a \emph{seed set}; each seed set is in a separate population that is disjoint from the test sets. The seed set $i$ corresponds 
to treatment version --agent action-- $A_i$. 
The seed set will be targeted with a   promotional campaign, 
and outcomes will be measured on units only in the test sets, 
say, number of purchases for each unit.
 The rationale is that the experimenter is interested in the
  viral marketing efficacy of the agents, i.e., their ability to
  select influential seed sets.

  Under interference, the treatment version (seed set) selected by
  agent $i$ induces a rate $\lambda_i$ on units assigned to $i$, and a
  rate $\gamma \lamc{i}$, where $0 \le \gamma \le 1$, on units
  assigned the other agent. The parameter $\gamma$ models the amount
  of interference; if $\gamma=0$ there is no interference, whereas
  $\gamma=1$ indicates maximum interference.  For the rest of this
  paper we will consider $\gamma$ known to agents and the designer,
  but this is without loss of generality.
  Rate $\lamc{i}$ can be interpreted as the rate that agent $i$ would
  achieve if the units that are targeted were its own units.
  Parameter $\gamma$ represents a discount because the targeted units
  are in the test set of another agent.

  The setting with interference is depicted in Figure~\ref{fig:m1}.
  The labels on the edges correspond to the effects from the seed
  sets, including interference effects.  For example, the purchase
  rate in test set 2 (units assigned to agent 2) is equal to $\gamma
  \lamc{1} + \lam{2}$; the first term is the discounted influence from
  the seed set of agent 1, and the second term is the influence from
  the seed set of agent 2. Agents are scored based on outcomes of
  units in their respective test sets. Therefore, an agent can also
  ``free-ride'' on the conversion rate that comes from the action of
  the other agent.  \drawMone{1.2}{2.5} }

\newexample{Example 3(d) -- Poisson outcomes with interference}{
Given the interference model of Example 3(c), 
the actions are $A_1 = (\lam{1}, \lamc{1})$, $A_2 = (\lam{2}, \lamc{2})$, and the observed outcomes on the units in the test sets have the following distributions:
\begin{align}
\label{eq:poisson_interference}
\Yobs{u1} \sim \mathrm{Pois}(\lam{1}+ \gamma \lamc{2}), \nonumber \\
\Yobs{u2} \sim \mathrm{Pois}(\lam{2} + \gamma \lamc{1}).
\end{align}

To derive the performance of an agent, say agent 1, we need 
to replace agent 2 with a replicate of agent 1, playing action $A_2 =(\lam{1}, \lamc{1})$. In this case, the induced rate 
on the units assigned to agent 1 is actually equal to $\lam{1} + \lamc{1}$ since, by definition of our interference model in Example 3(c), a rate is discounted only from a seed set of one agent 
to the test set of another agent. 
Thus, the performance of agent $i$ for action 
$\alpha_i= (\lam{i}, \lamc{i})$ is equal to 
\begin{align}
\label{eq:performance:interference}
\chi(\alpha_i) = \ExCond{Y_u(\Z, \A)}{\A=\alpha_i\ones, Z_u=i}
 = \lam{i} + \lamc{i}.
\end{align}
}
\vspace{-0.2in}

It can be seen, by inspection of Eq. \eqref{eq:poisson_interference},
that the outcomes of one unit depend on the action of the other agent.
For example, the outcomes $\Yobs{.1}$ on units assigned to agent 1
depend on action $\lam{1}$ of agent 1 as well as action $\lamc{2}$ of
agent 2. Hence, the observed outcomes for one agent carries
statistical information for the action of the other agent. This
information should be used in order to correctly estimate the agent
qualities, and then the agent of highest quality.

However, the estimation of qualities is not possible through outcomes
\eqref{eq:poisson_interference}, because 
there exist multiple action profiles
for which the observed outcomes are equally likely.  It follows that
there is no identifying statistic, and our theory (e.g., Theorem
\ref{theorem:general}) cannot be applied.  Furthermore, the
variance-stabilization transformations that were shown to give more
powerful designs in Example~3(b) do not work. This is illustrated in
the following example.

\newexample{Example 3(e). -- Poisson outcomes with interference}{
  Consider the setup of Example 3(c) and an experiment $\design$ with
  the usual score function $\phi_i(\Yobs{..}) =\Yobsbar{.i}$.  
As the number of experimental units grows, 
Eq. \eqref{eq:poisson_interference} result in the 
following asymptotics.
\begin{align}
\sqrt{\UperA} \left(\Yobsbar{.1} - 
(\lam{1} + \gamma \lamc{2}) \right) 
\tod \mathcal{N}(0, \lam{1} + \gamma \lamc{2}),
\nonumber \\
\sqrt{\UperA} \left(\Yobsbar{.2} - 
(\lam{2} + \gamma \lamc{1}) \right) 
\tod \mathcal{N}(0, \lam{2} + \gamma \lamc{1}). \nonumber
\end{align}

Therefore, the probability that agent 1 wins is
\begin{align}
P_1(\A | \design) = \mathrm{Pr}(\Yobsbar{.1} > \Yobsbar{.2})
 =\Phi\left(\sqrt{\UperA} 
\frac{(\lam{1}-\gamma \lamc{1}) - (\lam{2}-\gamma \lamc{2})}
{\sqrt{\lam{1} + \gamma \lamc{1} + \lam{2} + \gamma \lamc{2}}}     \right).
\end{align}

This design is not incentive-compatible because agent 1 prefers 
a large $\lam{1} - \gamma \lamc{1}$ and a small $\lam{1} + \gamma
\lamc{1}$.  As can been seen from Figure \ref{fig:m1}, a purchase rate
of $\gamma \lamc{1}$ from the seed set of agent 1 only benefits agent
2. Thus, agent 1 wants to benefit its assigned units (test set 1)
while minimizing the spillovers to test set 2 that benefit only agent
2.  However, the experimenter wants to know
something very different. In particular, given the definition
of performance in Example~3(d), the experimenter 
wants to know the maximum $\lam{1} + \lamc{1}$ that 
agent $1$ can achieve (and maximum $\lam{2} + \lamc{2}$, for agent 2). This quantity is of interest 
because it is the quantity that agent 1 would maximize if a copy 
of agent 1 substituted agent 2, and also played $(\lam{1}, \lamc{1})$.

Using the variance-stabilizing transformation of Example~3(b), does not solve the problem.
In particular, if we use $\phi_i(\Yobs{..})= 2\sqrt{\Yobsbar{.i}}$ as
the score function, then the winning probability of agent 1 becomes
\begin{align}
P_1(\A| \design)= \Phi\left(\sqrt{\UperA/2}
(\sqrt{\lam{1}+\gamma \lamc{2}}- \sqrt{\lam{2}+\gamma \lamc{1}})\right). \nonumber
\end{align}

The incentive problem remains because agent 1 still wants achieve a
high purchase rate $\lam{1}$ on units in test set 1, and a low rate
$\lamc{1}$ in units of test set 2.  }

\subsection{Dealing with strategic interference through better designs}
\label{section:within_interference:sufficient}

We now describe a method to construct an incentive-compatible design in the viral marketing problem with interference.  The idea is to introduce a new design that will provide an identifying statistic, and
then define appropriate score functions to fulfill the conditions of
Theorem~\ref{theorem:general} that guarantee incentive-compatibility.

\newexample{Example 3(f). -- Poisson outcomes with interference -- New
  design}{
We consider the following new design.  
 The units are split
  in two groups, say $G_1$ and $G_2$. Within each group, units are
randomly assigned to the two agents, resulting in 2 test sets per agent.  For example, group
  $G_1$ has two test sets, namely $G_{11}$ with units assigned to agent 1, and
  $G_{12}$ with units assigned to agent 2.  Similarly, group $G_2$ has test sets $G_{21}$ with units assigned to agent 1, and $G_{22}$ with units assigned to agent 2.  Test sets in the same group may be overlapping. 
In addition, each agent is free to pick one \emph{seed set}; each seed set is in a separate population that is disjoint from the test sets. The seed set $i$ corresponds 
to treatment version --agent action-- $A_i$.
The outcomes $Y$, say  number of purchases for each unit, for each agent $i$, will be measured on units only in their two test sets,
namely $G_{1i}$ and $G_{2i}$. This design is depicted
  in Figure \ref{fig:m2}.  \drawMtwo{1.3}{2.4}

The outcomes model is similar to the design of Example 3(c) (see also Figure \ref{fig:m1}).  A seed set $i$ --action $A_i$-- induces a rate $\lam{i}$ on
units of group $G_i$, and a rate $\lamc{i}$ on units of the other
group. The rate is assumed to be discounted when the seed set is
targeting units in a test set of another agent. For example, units in
test set $G_{12}$ will have purchase rate $\lamc{2} + \gamma \lam{1}$;
the rate $\lamc{2}$ originates from seed set 2 affecting units in
group $G_1$, and rate $\lam{1}$ is from seed set 1 affecting units in
$G_1$, discounted by $\gamma$ because $G_{12}$ is a test set of agent
2. Thus, action $A_i$ is associated with a pair of rates, 
$A_i = (\lam{i}, \lamc{i})$.

Agent 1's action is $A_1= (\lam{1}, \lamc{1})$, and agent $2$'s action is $A_2 = (\lam{2}, \lamc{2})$. 
Therefore, the observed outcomes of units are distributed as follows:
\begin{align}
\label{eq:interference_outcomes_Mtwo}
\Yobs{ui} \sim 
\begin{cases} 
		\mathrm{Pois}(\lam{1} + \gamma \lamc{2}), &\mbox{if } u \in G_{11}, \\ 
		\mathrm{Pois}(\lamc{2} + \gamma \lam{1}), &\mbox{if } u \in G_{12}, \\
		\mathrm{Pois}(\lamc{1} + \gamma \lam{2}), &\mbox{if } u \in G_{21}, \\ 
		\mathrm{Pois}(\lam{2}  + \gamma \lamc{1}), &\mbox{if } u \in G_{22}.
\end{cases}
\end{align}
}

Using the same interference model (parameter $\gamma$ of discounted
influence) introduced in Example 3(c), the new design of Figure
\ref{fig:m2} now provides more information about the agent actions,
and thus their performance, through outcomes
\eqref{eq:interference_outcomes_Mtwo}.  This additional information
provides an identifying statistic that can be used to define score
functions that make the design of Figure~\ref{fig:m2}
incentive-compatible.

\newexample{Example 3(g). -- Poisson outcomes}{ 
By symmetry of the new design, the experimenter is interested to estimate $\chi(A_i) = \lam{i} +  \lamc{i}$.
Let $\bar{Y}_{ij}$ be the sample mean
  of outcomes of units in test set $G_{ij}$, and let $Y =
  (\bar{Y}_{11}, \bar{Y}_{12}, \bar{Y}_{21}, \bar{Y}_{22})^\intercal$.
  Define the matrices
 \[
B =  \left( \begin{array}{cccc}
1 & 1  & 0 & 0 \\
0 & 0  & 1 & 1 \\
\end{array} \right), \text{ and }
C =  \left( \begin{array}{cccc}
1 & 0  & \gamma & 0 \\
\gamma & 0  & 1 & 0 \\
0 & 1   & 0 & \gamma \\
0 & \gamma  & 0 & 1 \\
\end{array} \right). \]

\newcommand{\Da}{D_{\A}}

Denote the action profile as $\A = (\lam{1}, \lamc{1}, \lamc{2}, \lam{2})^\intercal$. 
Further, let $\Da = \mathrm{diag}(C \A)$  be the diagonal 
matrix with diagonal elements from the vector $C \A$.
By Eq. \eqref{eq:interference_outcomes_Mtwo}, as the number 
of units grows, we have 
\begin{align}
\sqrt{\maxUnit/4} (Y - C \A) \tod \mathcal{N} (0, \Da).
\end{align}
The term $m/4$ is because there are $m/4$ units per test set.
Now define the statistic $T = B C^{-1} Y$.
Since $\chiA = (\lam{1} +\lamc{1}, \lam{2}+\lamc{2})^\intercal = B \A$, it holds, asymptotically,\footnote{
The normality of $T$ follows from normality of $Y$. 
The expected value of $T$ is $\Ex{T} = \Ex{B C^{-1} Y} = \Ex{B C^{-1} C \A} = B \A$, and 
its variance is $\Var{T} = \Var{B C^{-1} Y} = B C^{-1} \Var{Y} (C^{-1})^\intercal B^\intercal
 = B C^{-1} (\Da/m) (C^{-1})^\intercal B^\intercal$.
}
\begin{align}
\sqrt{m/4} (T - \chiA) \tod \mathcal{N}(0, B C^{-1} \Da (C^{-1})^\intercal B^\intercal).
\end{align}
Therefore, the new design has identifiable performance, 
and $T$ is an identifying statistic, with covariance matrix 
$\Sa = B C^{-1} \Da (C^{-1})^\intercal B^\intercal$.

Now, using notation of Theorem \ref{theorem:general}, define the score function simply as 
\begin{align}
\label{eq1:score}
\phi_i(\Yobs{..}) = f(T_i) = T_i.
\end{align}
Thus, the Jacobian of $\phi$ is $\Jacob=\mathbb{I}$, the identity matrix. The matrix $V(\A)$ of Theorem \ref{theorem:general} is
calculated as
\begin{align}
V(\A) = \Jacob \Sigma(\A) \Jacob^\intercal = 
B C^{-1} \Da (C^{-1})^\intercal B^\intercal.
\end{align}
Through simple but tedious matrix algebra we obtain,
\begin{align}
V(\A) = \frac{1}{(1-\gamma^2)^2} 
 \left( \begin{array}{cc}
d_1 + \gamma^2 d_2 + d_3 + \gamma^2 d_4 &  -\gamma \sum_{i=1}^4 d_i \\
-\gamma \sum_{i=1}^4 d_i & \gamma^2 d_1 + d_2 + \gamma^2 d_3 + d_4 \\
\end{array} \right),
\end{align}
where $(d_i)$ are the diagonal elements of $\Da$;
thus, $d_1 = \lam{1} + \gamma \lamc{2}$, 
$d_2 = \gamma \lam{1} + \lamc{2}$, 
$d_3 = \lamc{1}  + \gamma \lam{2}$,
and $d_4 = \gamma \lamc{1} + \lam{2} $.
In particular, 
\begin{align}
\sum_{i=1}^4 d_i = (1+\gamma) \left[(\lam{1} + \lamc{1})
 + (\lam{2} + \lamc{2})\right].
\end{align}
It follows from Theorem Eq. \eqref{eq:theorem:cov} of Theorem \ref{theorem:general}, 
\begin{align}
v^{ij}_f(\alpha | \A_{-i}) = & (d_1 + \gamma^2 d_2 + d_3 + \gamma^2 d_4)+ (\gamma^2 d_1 + d_2 + \gamma^2 d_3 + d_4)   - (-2 \gamma \sum_{i=1}^4 d_i) \nonumber \\
= & (1+\gamma)^2 \sum_{i=1}^4 d_i = (1+\gamma)^3
 \left[(\lam{1} + \lamc{1})
 + (\lam{2} + \lamc{2})\right]
 \nonumber,
\end{align}
if $i \ne j$, and 0 otherwise.
It follows that,
\begin{align}
\label{eq:1}
\arg \max_{\alpha_i \in \actions_i}   \left\{ \frac{f(\chi(\alpha_i))}
{ v^{ij}_f(\alpha_i|\A_{-i})^{1/2} }\right\}  \propto
\arg \max_{\alpha_i \in \actions_i}   \left\{ 
\frac{\lam{i} + \lamc{i}}
{\sqrt{(\lam{1} + \lamc{1})
 + (\lam{2} + \lamc{2}) } }\right\}.
\end{align}
}

The expression on the right of Eq. \eqref{eq:1} is increasing with
respect to $\chi(\alpha_i) = \lam{i} + \lamc{i}$.  Therefore, each
agent prefers to play actions $(\lam{i}, \lamc{i})$ so as to maximize
their sum, $\lam{i} + \lamc{i}$, which is the quantity of interest to
the experimenter. Condition \eqref{eq:theorem:general} of 
Theorem \ref{theorem:general}
is fulfilled. Thus, incentives are aligned under the new
design. Intuitively, the new design allows all agents to benefit 
from spillovers. For example, in the previous design, 
agent 1 could not benefit from the spillover of seed set 1 to 
test set 2, because agent 1's score was calculated 
only on test set 1. However, in the new design, the score 
of agent 1 includes outcomes from units in the test set $G_{21}$, 
which receives spillovers from seed set 1.

\section{Conclusion}
We introduced game theory into experiments where the treatments are determined by actions of strategic agents, and where treatments can interfere with each other. 
The goal of the experiment is to estimate the agent that is 
best with respect to a quantity of interest, defined 
in a context \emph{without} competition; e.g., average number 
of conversions from the agent's algorithm for viral marketing.
However, statistical estimation of the best agent
is based on experiment data, generated \emph{with} competition 
among agents. 
Thus, the game-theoretic setting poses new challenges to the statistical analysis of experiment data, and may often invalidate well-established experimental design methods.
%
The goal of incentive-compatible experimental design 
is to promote behaviors by agents that accord to the natural actions the agents would take in the experiment if there was no competition. 

When agent actions do not interfere with each other, we showed that incentive-compatible designs are possible through variance-stabilizing transformations of statistics
that estimate how agent would perform without competition.
Furthermore, we proved a result suggesting that variance
stabilization might, more generally, lead to more powerful
incentive-compatible experiment designs, in which better agents have
higher chances of winning.
In the presence of interference, we showed that 
more elaborate designs are generally necessary to 
obtain statistics that estimate agent performances.
In the context of a viral marketing application, 
we showed how a better design can be constructed that can 
account for interference among agents, e.g., 
when agents are able 
to free-ride on the advertising campaign of 
other agents. 

\small
\bibliographystyle{plain}
\bibliography{ic-exp-design-ec2015}

\begin{thebibliography}{10}

\bibitem{athey2008comparing}
Susan Athey, Jonathan Levin, and Enrique Seira.
\newblock Comparing open and sealed bid auctions: Evidence from timber
  auctions.
\newblock Technical report, National Bureau of Economic Research, 2008.

\bibitem{besag1986statistical}
Julian Besag and Rob Kempton.
\newblock Statistical analysis of field experiments using neighbouring plots.
\newblock {\em Biometrics}, pages 231--251, 1986.

\bibitem{bickel2001mathematical}
P.J. Bickel and K.A. Doksum.
\newblock {\em Mathematical Statistics: Basic Ideas and Selected Topics}.
\newblock Number v. 1 in Holden-Day series in probability and statistics.
  Prentice Hall, 2001.

\bibitem{box1978statistics}
George~EP Box, William~Gordon Hunter, J~Stuart Hunter, et~al.
\newblock Statistics for experimenters.
\newblock 1978.

\bibitem{cox1998delta}
C~Cox.
\newblock Delta method.
\newblock {\em Encyclopedia of biostatistics}, 1998.

\bibitem{cox2000theory}
David~Roxbee Cox and Nancy Reid.
\newblock {\em The theory of the design of experiments}.
\newblock CRC Press, 2000.

\bibitem{dash2001caveats}
Denver Dash and Marek Druzdzel.
\newblock Caveats for causal reasoning with equilibrium models.
\newblock In {\em Symbolic and Quantitative Approaches to Reasoning with
  Uncertainty}, pages 192--203. Springer, 2001.

\bibitem{david1996designs}
Olivier David and Rob~A Kempton.
\newblock Designs for interference.
\newblock {\em Biometrics}, pages 597--606, 1996.

\bibitem{pearl2000causality}
Judea Pearl.
\newblock {\em Causality: models, reasoning and inference}, volume~29.
\newblock Cambridge Univ Press, 2000.

\bibitem{rubin1980comment}
Donald~B Rubin.
\newblock Comment.
\newblock {\em Journal of the American Statistical Association},
  75(371):591--593, 1980.

\bibitem{toulis2013estimation}
Panos Toulis and Edward Kao.
\newblock Estimation of causal peer influence effects.
\newblock In {\em Proceedings of The 30th International Conference on Machine
  Learning}, pages 1489--1497, 2013.

\bibitem{toulis2015statistical}
Panos Toulis and David~C Parkes.
\newblock Long-term causal effects of interventions in multiagent economic
  mechanisms.
\newblock {\em arXiv preprint arXiv:1501.02315}, 2015.

\bibitem{toulisCODE14}
Panos Toulis, David~C. Parkes, Elery Pfeffer, James Zou, and Guy Gildor.
\newblock Incentive-compatible experiment design (extended abstract).
\newblock In {\em Conference on Digital Experimentation {(CODE@MIT, 2014)}},
  2014.

\end{thebibliography}



\medskip
\section*{Appendix} 
\appendix 

\section{Extension to multiple blocks}
\label{section:multiple_blocks}
 
In this paper, our theory is developed and applied 
assuming only one block.
However, it is straightforward to extend it to multiple 
blocks in a typical blocking experiment design.
In this section, we give an outline of this extension.

 The treatment assignment rule $\psi$ now groups units into $B$ blocks
based on their covariates, and then randomizes treatment (i.e.,
the assignment of units to agents) 
within the
blocks; blocking is performed in a deterministic way based on the
publicly known covariates $\{X_u\}$, for each unit $\unit$.  Formally, rule
$\psi$
is a probability distribution over the space of pairs of binary
matrices $\Psi \defeq (\{0, 1\}^{\maxUnit \times \maxBlock}, \{0,
1\}^{\maxUnit \times \maxAgent})$.

A pair $(W, Z) \in \Psi$ is called a \emph{treatment assignment}, and
has the following interpretation.  The element $W_{\unit\block} = 1$
if unit \unit\ is assigned to block \block, and it is 0 otherwise.
Similarly, $Z_{\unit\agent}=1$ if unit $\unit$ is assigned to agent
$i$, and it is 0 otherwise.  Using dot-notation 
$W_{.\block}$ is the $\block$th column of matrix $W$, $W_{\unit.}$ is
the $\unit$th row of $W$ as a $\maxBlock \times 1$ vector, and $W_{..}
\equiv W$. Similarly for $Z$ and other matrices. 
Finally the notation $(W, Z) \sim \psi$ will denote a
treatment assignment $(W, Z) \in \Psi$, that is sampled according to
rule $\psi$. 

\newexample{Example A1}{Consider four experimental units (consumers)
  and two treatments (marketing agents) that an experimenter wishes to
  evaluate.  In particular, the experimenter is interested to estimate
  which agent  can achieve the highest number of sales.
  Suppose that, for each unit $\unit$, the experimenter and the agents
  know the marriage status (only covariate). We assume that units
$\{1,2\}$ are not married and $\{3,4\}$ are, and these correspond
to the two blocks $b\in\{1,2\}$.
The experimenter suspects that the outcomes 
  differ systematically based on marriage status,
and  randomizes treatment within blocks.
This design corresponds to treatment  assignment rule 
$\psi$ which samples with equal probability 1/4
  from the treatment assignments $\{W, Z\}$ where $Z \in \left \{
    \left( \begin{array}{cc}
        1 & 0 \\
        0 & 1\\
        1 & 0 \\
        0 & 1\end{array} \right), \left( \begin{array}{cc}
        0 & 1 \\
        1 & 0\\
        1 & 0 \\
        0 & 1\end{array} \right), \left( \begin{array}{cc}
        1 & 0 \\
        0 & 1\\
        0 & 1 \\
        1 & 0\end{array} \right), \left( \begin{array}{cc}
        0 & 1 \\
        1 & 0\\
        0 & 1 \\
        1 & 0\end{array} \right) \right \}$ and $W=\left(
    \begin{array}{cc}
      1 & 0 \\
      1 & 0 \\
      0 & 1 \\
      0 & 1\end{array} \right)$ is the matrix that indicates the
  blocking.
Some examples of dot-notation follow: $W_{1.} = (1
  \doublespace 0)^\intercal$ is the assignment of unit $\unit$ over
  blocks, $W_{.2} = (0 \doublespace 0\doublespace 1 \doublespace
  1)^\intercal$ is the assignment over units in block 2, etc.}

With multiple blocks, agents are allowed to play 
different actions across blocks. We would thus write 
$A_{ib}$ for the action of agent $i$ in block $b$, and 
$\actions_{ib}$ for the action space of this action.

With multiple blocks, there is also an additional block index 
for the potential and observed outcomes. 
For example, $\Yobs{ubi}$ is now the observed outcome 
of unit $u$ assigned to block $b$ and agent $i$;
with dot-notation, $\Yobs{.b.}$ denotes the observed outcomes 
of units in block $b$.
The experiment design $\design$ has now multiple score functions, $\phi_{b}$, one per block.
For example, $\phi_{ib}(\Yobs{.b.})$ is the score 
of agent $i$ in block $b$ with data $\Yobs{.b.}$.
Similar extensions are straightforward 
for the concepts of performance, natural action, and quality.

Given block-specific score functions, the winner 
of the experiment is the agent who won the majority of blocks,
ignoring ties. When there is no interference across and within-blocks,
then the experimenter can design an incentive-compatible
design within each block using Theorem \ref{theorem:general}.
In this case, each block would have a separate identifying 
statistic. When the action space of an agent 
is the product space of the block action spaces, 
the agent will prefer to maximize its winning probability 
within each block. Therefore, the incentive-compatibility 
results of Theorems \ref{theorem:general} and \ref{theorem:no_interference} can be readily applied.
The same results can be applied in the problem with interference, 
assuming that there is no between-block interference, 
i.e., an action of agent $i$ in block $b$ does not 
affect the outcomes for agent $j$ in some other block $b'$.

\section{Proofs}
\begin{customthm}{\ref{theorem:general}}
\label{theorem:general_proof}
Fix agent actions $\A$, and consider design $\design=\designTuple$ that has an 
identifying statistic $T$ with covariance matrix $\Sa$. 
Let $\phi_i(\Yobs{..}) =f(T_i)$ for some function $f : \Reals \to 
\Reals$, and let $v_{ij}(\A)$ be the $ij$th element 
of $V(\A)$ defined in Eq. \eqref{eq:Va}.
Also define,
\begin{align}
v^{ij}_f(\alpha|\A_{-i}) = v_{ii}(\alpha, \A_{-i}) + v_{jj}(\alpha, \A_{-i}) - 
v_{ij}(\alpha, \A_{-i}) - v_{ji}(\alpha, \A_{-i}). \nonumber
\end{align}
The design $\design$ is incentive-compatible, if, for every agent $i$,
\begin{align}
\arg \max_{\alpha_i \in \actions_i}   \left\{ \frac{f(\chi(\alpha_i))}
{ v^{ij}_f(\alpha_i|\A_{-i})^{1/2}} \right\}  & = \arg \max_{\alpha_i \in \actions_i}   \left\{ \chi(\alpha_i)\right\} \defeq A_i^\star, \nonumber
\end{align}
for every agent $j$, and all actions $\A_{-i}$. In such case, we say that $T$ is aligned with performance $\chi$ through score $\phi$.
\end{customthm}
\begin{proof}
For a vector $x \in \Reals^{\maxAgent}$, let $f(x) = (f(x_1), f(x_2), \ldots,
f(x_\maxAgent))^\intercal$. From the Delta theorem \cite{bickel2001mathematical, cox1998delta}, and the asymptotic property \eqref{eq:identifying} of the identifying statistic $T$, we obtain
\begin{align}
\label{thm3:eq1}
\sqrt{\UperA}\left(f(T) - f(\chiA)\right) \tod \mathcal{N}(0, \Jacob \Sa \Jacob^\intercal),
\end{align}
where $\Jacob$ is the Jacobian of $f$ at $\chiA$ (by definition, 
this is a diagonal matrix). The probability that 
agent $i$ wins over $j$ is equal to 
\begin{align}
\label{thm3:eq2}
\mathrm{Pr}\left(\phi_i(\Yobs{..}) > \phi_j(\Yobs{..})\right) =
\mathrm{Pr}\left(c^\intercal f(T) > 0\right),
\end{align}
where $c = (0, \ldots, 1, 0, \ldots, -1, 0, \ldots)^\intercal$, is a 
$\maxAgent \times 1$ vector, with zero elements, except for $c_i=1$ and $c_j=-1$. 
Using Eq. \eqref{thm3:eq1}, we have
\begin{align}
\label{thm3:eq3}
\sqrt{\UperA}\left(c^\intercal f(T) - c^\intercal f(\chiA)\right) \tod 
\mathcal{N}(0, c^\intercal\Jacob \Sa \Jacob^\intercal c). 
\end{align}
From \eqref{thm3:eq3}, probability \eqref{thm3:eq2} becomes
\begin{align}
\mathrm{Pr}\left(\phi_i(\Yobs{..}) > \phi_j(\Yobs{..})\right) =
\Phi\left( \frac{f_i(\chiA) - f_j(\chiA)}{v^{ij}_f(\A)^{1/2}}\right) = 
\Phi\left( \frac{\chi(A_i) - \chi(A_j)}{v^{ij}_f(\A)^{1/2}}\right), \nonumber
\end{align}
where $v^{ij}_f(\A)$ is given in Eq. \eqref{eq:theorem:cov}.
Therefore, agent $i$ 
maximizes its winning probability by playing the natural action, 
by property \eqref{eq:theorem:general}.
\end{proof}

\begin{customthm}{\ref{theorem:no_interference}}
Consider design $\design = \designTuple$ with an identifying statistic $T$ with covariance matrix $\Sa$.
Suppose Assumption \ref{assumption:no_interference} holds. If, for every agent $i$,
\begin{align}
 & \phi_i(\Yobs{..}) \equiv f(T_i), \text{ where } f:\Reals\to\Reals,  \nonumber \\ 
 & \Var{\phi_i(\Yobs{..})}  = \mathrm{const.},\nonumber \\
 & \arg \max_{\alpha_i \in \actions_i} f(\chi(\alpha_i)) = \arg \max_{\alpha_i \in \actions_i} \{\chi(\alpha_i)\} \defeq A_i^\star,\nonumber
\end{align}
then design $\design$ is incentive-compatible.
\end{customthm}
\begin{proof}
By Assumption \ref{assumption:no_interference} (no interference), $\Sa$ is diagonal; let $\Sa = \mathrm{diag}(\sigma_{ii}^2(\A))$.
Then, from Theorem \eqref{theorem:no_interference} and 
Condition \eqref{eq:condition:simple}, 
\begin{align}
\Var{\phi_i(\Yobs{..})} = f'(\chi(A_i))^2 \sigma_{ii}^2(\A) = c, \nonumber
\end{align}
for some constant $c>0$. Also by Condition \eqref{eq:condition:simple}, the Jacobian of $\phi$ 
at $\A$, is given by
$\Jacob = \mathrm{diag}(f'(\chi(A_i)))$. 
Using the notation of Theorem \ref{theorem:general},
\begin{align}
V(\A) = \Jacob \Sigma(\A) \Jacob^\intercal = \mathrm{diag}(f'(\chi(A_i))^2 \sigma_{ii}^2(\A)) = c \mathbb{I}.\nonumber
\end{align}
It follows, $v_f^{ij}(\alpha|\A_{-i}) = 2c$ for any $i, j$, where $v_f^{ij}$ 
is defined in Eq. \eqref{eq:theorem:cov}, Theorem \ref{theorem:general}.
Using Condition \eqref{eq:condition:monotone},
\begin{align}
\arg \max_{\alpha_i \in \actions_i}   \left\{ \frac{f(\chi(\alpha_i))}
{ v^{ij}_f(\alpha_i|\A_{-i})^{1/2}} \right\} =
  (1/2c)  \arg \max_{\alpha_i \in \actions_i} \left\{\chi(\alpha_i) \right\}
= A_i^\star\nonumber.
\end{align}
Thus, all conditions of Theorem \ref{theorem:general} are
fulfilled, and the design $\design$ is incentive-compatible.
\end{proof}

\begin{customthm}{\ref{theorem:rao}}
Consider an incentive-compatible design $\design = \designTuple$, 
where action sets $\actions_i \subseteq \Reals$ are compact, 
and performance $\chi$ is one-to-one and continuous.
Let,
\begin{align}
\sqrt{\UperA} \left(\phi_i(\Yobs{..})-\chi(A_i)\right) \tod
\mathcal{N}(0, \sigma^2(A_i)), \nonumber
\end{align}
where function $\sigma^2 : \actions \to \Reals^{+}$ satisfies
\begin{align}
\chi(\alpha_i') \ge \chi(\alpha_i) \Rightarrow \sigma^2(\alpha_i') \ge \sigma^2(\alpha_i), \nonumber
\end{align}
for every agent $i$, and all actions $\alpha_i', \alpha_i \in \actions_i$.
Consider a design $\design' = (\psi, \phi')$,
where $\phi'_i(\Yobs{..}) = \nu(\phi_i(\Yobs{..}))$, for each agent $i$, with
$\nu(\cdot)$ defined by
\begin{align}
\nu(y) = \int^y \frac{1}{\sqrt{\sigma^2(\chi^{-1}(z)})} dz. \nonumber
\end{align}
Then, design $\design'$ is incentive-compatible and more powerful than $\design$, if $\nu(\cdot)$ is convex,
or $1/\sqrt{\sigma^2(\chi^{-1}(\cdot))}$ and 
$\sigma^2(\chi^{-1}(\cdot))$ are both convex.
\end{customthm}
\begin{proof}
From the univariate Delta theorem,
\begin{align}
\sqrt{\UperA}\left(\nu(\phi_i(\Yobs{..}) - \nu(\chi(A_i))\right)
\tod \mathcal{N}(0, 1), \nonumber
\end{align}
since $\nu'(\chi(A_i))^2 \sigma^2(A_i) = 1$, by Eq. \eqref{eq:rao:nu}. 
For brevity, set $\chi(A_i) \defeq \chi_i$ and $\sigma^2(A_i) \defeq \sigma_i^2$. Without loss of generality, assume $\chi_i \ge \chi_j $. 
The probability that agent $i$ wins over agent $j$ 
in design $\design'$ is equal to,
\begin{align}
P_1(\A|\design') = \Phi\left(\sqrt{\UperA/2}( \nu(\chi_i) - \nu(\chi_j)) \right). \nonumber 
\end{align} 
In the old design, $\design$, this probability 
is equal to 
\begin{align}
P_1(\A|\design) = \Phi\left(\sqrt{\UperA}
\frac{\chi_i - \chi_j}{\sqrt{\sigma_i^2 + \sigma_j^2}} \right). \nonumber 
\end{align} 

\vspace{5pt} 
\noindent \emph{Case 1 -- Convex $\nu(\cdot)$.}  
By convexity of $\nu$ we have
\begin{align}
\label{thm4:eq1}
\frac{\nu(\chi_i) - \nu(\chi_j)}{\chi_i - \chi_j} \ge \nu'(\chi_j).
\end{align}
By definition \eqref{eq:rao},  $\nu'(\chi_j)^2 \sigma_j^2=1$.
By property \eqref{eq:powerful:var}, $\sigma_i^2 \ge \sigma_j^2$
since $\chi_i \ge \chi_j$. Hence, $\nu'(\chi_i)^2 \sigma_i^2 =1 
\Rightarrow \nu'(\chi_i)^2 \le \nu'(\chi_j)^2$. It follows, 
\begin{align}
\label{thm4:eq2}
\nu'(\chi_j)^2 \sigma_j^2 + \nu'(\chi_j)^2 \sigma_i^2 \ge 2 \nonumber
\Rightarrow \\
 \nu'(\chi_j) \ge \sqrt{\frac{2}{\sigma_i^2 + \sigma_j^2}}.
\end{align}
Combining \eqref{thm4:eq1} and \eqref{thm4:eq2}, 
we obtain
\begin{align}
 \frac{\nu(\chi_i) - \nu(\chi_j)}{\sqrt{2}} \ge \frac{\chi_i - \chi_j}{\sqrt{\sigma_i^2 + \sigma_j^2}}
\Rightarrow  \Phi\left(\sqrt{\UperA/2}( \nu(\chi_i) - \nu(\chi_j)) \right)
 \ge  \Phi\left(\sqrt{\UperA}
\frac{\chi_i - \chi_j}{\sqrt{\sigma_i^2 + \sigma_j^2}} \right), \nonumber
\end{align}
which implies that design $\design'$ is more powerful than $\design$.

\vspace{5pt} 
\noindent \emph{Case 2 -- $1/\sqrt{\sigma^2(\chi^{-1}(\cdot))}$ and 
$\sigma^2(\chi^{-1}(\cdot))$ are both convex.}  
It holds, 
\begin{align}
\frac{\nu(\chi_i) - \nu(\chi_j)}{\chi_i - \chi_j} = 
\frac{1}{\chi_i - \chi_j} \int_{\chi_j}^{\chi_i} \frac{1}{\sqrt{\sigma^2(\chi^{-1}(z))}} dz \ge  
\frac{1}{\sqrt{\sigma^2(\chi^{-1}((\chi_j + \chi_i)/2))}}
\nonumber \\ 
\ge \frac{1}{\sqrt{\sigma^2(\chi^{-1}(\chi_j))/2 
+\sigma^2(\chi^{-1}(\chi_i))/2}}\defeq \sqrt{\frac{2}{\sigma_i^2 + \sigma_j^2}} \nonumber.
\end{align}
The first inequalty is obtained by convexity of $1/\sqrt{\sigma^2(\chi^{-1}(\cdot))}$, 
and the second by convexity of $\sigma^2(\chi^{-1}(\cdot))$.
To finish the proof we follow the same arguments as in Case 1.
\end{proof}


\section{Remarks on variance stabilization}
\label{section:varstab}
In Theorem \ref{theorem:no_interference}, the 
variance of the score functions $\phi_i$ is stabilized (made constant)
through a transformation $f$. 
Such transformations that stabilize the variance of a statistic, 
 are known as \emph{variance-stabilizing}
transformations in statistics, and they are of fundamental importance in various tasks, such as hypothesis testing and estimation. For
example, consider a sample average of $n$ independent Poisson random variables with mean $\lambda$.
The asymptotic distribution of the sample average is 
$\bar{Y} \sim \mathrm{Poisson}(\lambda/n)$. 
In the limit, $\sqrt{n}(\bar{Y}- \lambda) \tod \m{N}(0, \lambda)$. 
This asymptotic result is not useful to 
construct a confidence interval for the unknown 
parameter $\lambda$ because the variance of $\bar{Y}$ depends on that unknown parameter.
However, through the Delta theorem, 
$2 \sqrt{n}(\sqrt{\bar{Y}} -\sqrt{\lambda}) \tod
\m{N}(0, 1)$ i.e., the variance of $\sqrt{\bar{Y}}$ is 
constant; the statistic $\sqrt{\bar{Y}}$ can be used to obtain \emph{exact} confidence
intervals for $\lambda$. 

In our setting, the variance stabilization helps to mitigate the
risk-return trade-off that strategic agents can undertake in an
experiment. Loosely speaking, when the variance is stabilized a worse
agent cannot benefit by being more risky, and a better agent cannot
benefit by being more conservative. Rather, incentives are aligned
such that every agent will do its best, assuming
the remaining conditions of Theorem \ref{theorem:no_interference} are fulfilled.

\section{Discussion}
\label{section:discussion}
Our approach to design incentive-compatible experiments
has been through the use of an identifying statistic, 
i.e., a statistic that can estimate the agent performances without competition. In many situations, such a statistic exists, e.g., by using
sample summaries (means, variances, etc), and then appealing to the
central limit theorem.  In most realistic cases, 
a key assumption will be that the outcomes have a known parametric form. In this paper, we made such parametric assumptions in 
our viral marketing example.

However, an experimenter might be unwilling to make such parametric
modeling assumptions. An alternative would then 
be either to use a nonparametric test for the quantities of interest (i.e., agent 
performances), or a randomization-based analysis. The former includes 
a wide-class of nonparametric methods, and we  plan to investigate
it in future work. It should be noted, however, that even nonparametric tests
have crucial underlying assumptions, e.g., exchangeability of observed data, that are not easy to validate. 
In many situations, such assumptions are more critical than, for example, 
normality assumptions that can be quite robust under many scenarios
\cite[Appendix 3A]{box1978statistics}. The latter method of 
randomization-based analysis usually starts from a null hypothesis 
which aims to provide evidence for the likelihood of certain observed quantities, e.g., through p-values. However, it is hard to 
test such hypotheses in our setting because agents can freely choose the versions 
of the treatment to apply. Therefore, one cannot use the null hypothesis 
to \emph{impute} counterfactuals, i.e., outcomes that would have been observed under a different randomization because agents 
act in a strategic, non-random way.

In the case with interference, the assumption 
that an identifying statistic exists has two components.
First, it is required that the experimenter has a good idea about 
the \emph{model} of interference, e.g., that an 
agent action affects the outcomes for another agent linearly, 
as in Example 3(c).
Assumptions on the model of interference are frequent in practice because they help to deal with interference after the experiment has been performed
\cite{besag1986statistical}.
Second, it is required that the experimenter knows exactly the hyperparameters 
of the assumed interference model. In the viral 
marketing problem of Section \ref{section:interference}, 
a scalar parameter $\gamma$ was used to model 
interference. In our examples, we assumed that $\gamma$ was known.
One way to avoid this problem is 
to treat such parameters of interference as \emph{nuisance} parameters, and then use a suitable statistical method; e.g., use profile likelihood 
instead of the true, but unknown, likelihood to obtain proxies for the 
maximum-likelihood estimates. A Bayesian approach would be to set priors for such parameters and then obtain a posterior predictive distribution for the unknown agent performances. 
Agents would then be scored according to this 
posterior distribution, but this would not alter the core 
of our methodology.

\end{document}